\documentclass[journal]{IEEEtran}%[(draftcls,onecolumn),draftclsnotoot,twosides]
\ifCLASSINFOpdf
\else
\fi
\usepackage{mathrsfs}
\usepackage{amsmath}
\usepackage{amsfonts}
\usepackage{amsthm}
\usepackage{amssymb}
\usepackage{graphicx}
\usepackage{epstopdf}
\usepackage{indentfirst}
\usepackage{array}
\usepackage{cite}
\usepackage{enumerate}
\usepackage[linesnumbered,ruled,vlined]{algorithm2e}
\usepackage{multirow}
\usepackage{verbatim}
\usepackage{stfloats}
\begin{document}
%
% paper title
% can use linebreaks \\ within to get better formatting as desired
\title{Rate-Compatible Punctured Polar Codes: Optimal Construction Based on Polar Spectra}

% author names and affiliations
% use a multiple column layout for up to three different
% affiliations
\author{Kai Niu, \IEEEmembership{Member,~IEEE}, Jincheng Dai, \IEEEmembership{Student Member,~IEEE}, Kai Chen, \IEEEmembership{Student Member,~IEEE}, Jiaru Lin,  \IEEEmembership{Member,~IEEE}, Q. T. Zhang,  \IEEEmembership{Fellow,~IEEE}, Athanasios V. Vasilakos, \IEEEmembership{Senior Member,~IEEE}
\thanks
{
This work is supported by the National Natural Science Foundation of China (No. 61171099), National High-tech R\&D Program (863 Program) (No. 2015AA01A709) and Qualcomm Corporation. The material in this paper was presented in part at the IEEE International Conference on Communications, Budapest, Hungary, June 2013. \protect\\
\indent K. Niu, J. Dai, K. Chen and J. R. Lin are with the Key Laboratory of Universal Wireless Communications, Ministry of Education,
Beijing University of Posts and Telecommunications, Beijing, 100876, China (e-mail: \{niukai, daijincheng, kaichen, jrlin\}@bupt.edu.cn). \protect\\
\indent Q. T. Zhang was with the Department of Electronic Engineering, City University of Hong Kong (e-mail: qtzhang@ieee.org). \protect\\
\indent Athanasios V. Vasilakos is with Department of Computer Science, Electrical and Space Engineering, Lule\aa\ University of Technology, SE-931 87 Skellefte\aa, Sweden, (e-mail: th.vasilakos@gmail.com).
}}

% conference papers do not typically use \thanks and this command
% is locked out in conference mode. If really needed, such as for
% the acknowledgment of grants, issue a \IEEEoverridecommandlockouts
% after \documentclass

% use for special paper notices
%\IEEEspecialpapernotice{(Invited Paper)}

% make the title area
\maketitle

\newtheorem{theorem}{Theorem}
\newtheorem{example}{Example}
\newtheorem{corollary}{Corollary}
\newtheorem{lemma}{Lemma}
\newtheorem{proposition}{Proposition}
\newtheorem{definition}{Definition}
\newtheorem{remark}{Remark}

\begin{abstract}
%\boldmath
Polar codes are the first class of constructive channel codes achieving the symmetric capacity of the binary-input discrete memoryless channels. But the corresponding code length is limited to the power of two. In this paper, we establish a systematic framework to design the rate-compatible punctured polar (RCPP) codes with arbitrary code length. A new theoretic tool, called polar spectra, is proposed to count the number of paths on the code tree with the same number of zeros or ones respectively. Furthermore, a spectrum distance SD0 (SD1) and a joint spectrum distance (JSD) are presented as performance criteria to optimize the puncturing tables. For the capacity-zero puncturing mode (punctured bits are unknown to the decoder), we propose a quasi-uniform puncturing algorithm, analyze the number of equivalent puncturings and prove that this scheme can maximize SD1 and JSD. Similarly, for the capacity-one mode (punctured bits are known to the decoder), we also devise a reversal quasi-uniform puncturing scheme and prove that it has the maximum SD0 and JSD. Both schemes have a universal puncturing table without any exhausted search. These optimal RCPP codes outperform the performance of turbo codes in LTE wireless communication systems.
\end{abstract}
% IEEEtran.cls defaults to using nonbold math in the Abstract.
% This preserves the distinction between vectors and scalars. However,
% if the conference you are submitting to favors bold math in the abstract,
% then you can use LaTeX's standard command \boldmath at the very start
% of the abstract to achieve this. Many IEEE journals/conferences frown on
% math in the abstract anyway.


% For peer review papers, you can put extra information on the cover
% page as needed:
% \ifCLASSOPTIONpeerreview
% \begin{center} \bfseries EDICS Category: 3-BBND \end{center}
% \fi
%
% For peerreview papers, this IEEEtran command inserts a page break and
% creates the second title. It will be ignored for other modes.
\IEEEpeerreviewmaketitle

\section{Introduction}
% no \IEEEPARstart \lettrine[lines=2]{P}
Rate-compatible coding schemes are desirable to provide different error protection requirements, or accommodate time-varying channel characteristics. Especially, we would like to design a pair of encoder and decoder which can adapt both different code length and different code rate without changing their basic structure in the hybrid automatic repeat-request (HARQ) protocols. In such cases, rate compatible punctured convolutional (RCPC) codes \cite{RCPC_Hagenauer} or rate compatible punctured turbo (RCPT) codes \cite{RCPT_Rowitch} are typical coding techniques, which are broadly applied in modern wireless communication systems, such as LTE (Long Term Evolution). Recently, as the first constructive capacity-achieving coding scheme, polar codes \cite{Polarcode_Arikan} reveal the advantages of error performance and many attractive application prospects. According to the original code construction \cite{Polarcode_Arikan}, polar codes are also able to support rate compatibility partially since the code rate can be precisely adjusted by adding or deleting one information bit. However, the code length $N$ still is limited to the power of two, i.e., $N=2^n$. Consequently, puncturing code bits and shortening the code length becomes the key technique of designing good rate-compatible punctured polar (RCPP) codes.

To the best of the authors' knowledge, the puncturing schemes of polar codes can be summarized as two categories. First, some code bits are punctured in the encoder and the decoder has no \emph{a priori} information about these bits which can be regarded as the ones transmitting over zero-capacity channels. In this paper, we call this category as the capacity-zero (C0) puncturing mode. Second, the values of the punctured code bits are predetermined and known by the encoder and decoder. Thus the associated channels can be regarded as one-capacity channels. We use the capacity-one (C1) puncturing mode to sketch the feature of this category.

For the puncturing schemes under the C0 mode, Eslami \emph{et al.} first proposed a stopping-tree puncturing to match arbitrary code length under the belief propagation (BP) decoding \cite{Practical_Eslami,Finite_Eslami}. Then, Shin \emph{et al.} proposed a reduced generator matrix method to efficiently improve the error performance of the RCPP codes under the successive cancellation (SC) decoding \cite{LCPC_Shin}, whereas searching the good polarizing matrices is still a time consuming process. In \cite{PuncPattern_Zhang}, a heuristic puncturing approach was proposed for the codes with short length. In \cite{RCPP_Niu}, an efficiently universal puncturing scheme, named quasi-uniform puncturing algorithm (QUP) was proposed and the corresponding RCPP codes can outperform the performance of turbo codes in 3G/4G wireless systems.

On the other hand, for the puncturing schemes under the C1 mode, Wang \emph{et al.} \cite{Novel_Punc} first introduced the concept of capacity-one puncturing and devised a simple puncturing method by finding columns with weight 1 to improve the error performance of SC decoding. Later, the author in \cite{Shorten_PC} exploited the structure of polar codes and proposed a reduced-complexity search algorithm to jointly optimize the puncturing patterns and the values of the punctured bits.

To sum up, for the mainstream SC/SC-like decoding, most of the current puncturing schemes under the C0 or C1 modes are heuristic methods and lack of a systematic framework to design the RCPP codes. Intuitively, the optimal punctured scheme under the SC decoding can be obtained by enumerating each punctured pattern and calculating the relative upper bound of block error rate (BLER). Obviously, this exhausted search is intractable due to the prohibitive complexity. Theoretically, like the optimization of RCPC or RCPT codes, RCPP codes can also be constructed by the optimization of the distance spectra (DS) or weight enumeration function (WEF) \cite{Book_Lin} for different punctured patterns. But due to the high complexity of DS/WEF calculation of polar codes \cite{Distance_spectrum,Weight_distribution}, it is also unrealistic to design RCPP codes based on these metrics. Hence, designing a feasible and computable measurement is crucial for the optimization of RCPP codes under the SC decoding.

In this paper, we establish a complete framework to design and optimize the RCPP codes under the SC/SC-like decoding. Based on this framework, we obtain the optimal puncturing schemes for both modes. The main contributions of this paper can be summarized as follows.

(1) First, we propose a new tool, called polar spectra (PS), to simplify the performance evaluation of RCPP codes under SC decoding. Conceptually, polar spectra are defined on the code tree and include two categories: PS1 and PS0, which represent the number of paths with the same Hamming weight or complemental Hamming weight (the number of zeros) respectively.

Based on PS, we introduce two kinds of path weight enumeration function (PWEF1 and PWEF0) to indicate the distribution of (complemental) path weight. Furthermore, three performance metrics, the spectrum distance for PWEF0 (SD0), the spectrum distance for PWEF1 (SD1), and joint spectrum distance (JSD) for the entire PS, are defined to optimize the distribution of (complemental) path weight under two puncturing modes (C0 and C1).

(2) Second, for the C0 mode, thanks to the easily analyzed property of PS, we prove that the quasi-uniform puncturing (QUP) algorithm proposed in \cite{RCPP_Niu} can maximize the metrics SD1 and JSD. Moreover, we analyze the structure feature of this puncturing and obtain the exact number of equivalent puncturing tables.

(3) Third, for the C1 mode, we propose a new reversal quasi-uniform puncturing (RQUP) and prove that this scheme can maximize the metrics SD0 and JSD.

The remainder of the paper is organized as follows. Section \ref{section_II} describes the preliminaries of polar codes, including polar coding, decoding algorithm, and upper bounds analysis of Bhattacharyya parameter. Section \ref{section_III} describes the puncturing modes of RCPP codes and sketches out the en-/decoding process. The concepts of polar spectra, PWEFs (PWEF0 and PWEF1), and spectrum distances (SD0, SD1, and JSD) are introduced in Section \ref{section_IV}. The QUP algorithm is presented and proved to be the optimal one under the C0 puncturing mode in Section \ref{section_V}. Similarly, the RQUP scheme is proposed and proved to maximize the SD0 and JSD under the C1 puncturing mode in Section \ref{section_VI}. Section \ref{section_VII} provides the numerical analysis for various puncturing schemes and simulation results for RCPP and turbo codes in LTE systems. Finally, Section \ref{section_VIII} concludes the paper.

\section{Preliminary of Polar Codes}
\label{section_II}
\subsection{Notation Conventions}
In this paper, calligraphy letters, such as $\mathcal{X}$ and $\mathcal{Y}$, are mainly used to denote sets, and the cardinality of $\mathcal{X}$ is defined as $\left|\mathcal{X}\right|$. The Cartesian product of $\mathcal{X}$ and $\mathcal{Y}$ is written as $\mathcal{X}\times \mathcal{Y}$ and $\mathcal{X}^n$ denotes the $n$-th Cartesian power of $\mathcal{X}$.

We write $v_1^N$ to denote an $N$-dimensional vector $\left(v_1,v_2,\cdots,v_N\right)$ and $v_i^j$ to denote a subvector $\left(v_i,v_{i+1},\cdots,v_{j-1},v_j\right)$ of $v_1^N$, $1\leq i,j \leq N$. Further, given an index set $\mathcal{A}\subseteq \mathcal{I}=\{1,2,\cdots,N\}$ and its complement set $\mathcal{A}^c$, we write $v_\mathcal{A}$ and $v_{\mathcal{A}^c}$ to denote two complementary subvectors of $v_1^N$, which consist of $v_i$s with $i\in\mathcal{A}$ or $i\in\mathcal{A}^c$ respectively. We use $\mathbb{E}(\cdot)$ to denote the expectation operation of a random variable.

Throughout this paper, $\log$ means ``logarithm to base 2", and $\ln$ stands for the natural logarithm.
%We $\mathscr{X}$ $\mathfrak{X}$ $\mathbb{X}$
\subsection{Encoding and Decoding of Polar Codes}
Given a B-DMC $W: \mathcal{X}\to \mathcal{Y}$ with input alphabet $\mathcal{X}= \{0,1\}$ and output alphabet $\mathcal{Y}$, the channel transition probabilities can be defined as $W(y|x)$, $x\in \mathcal{X}$ and $y\in \mathcal{Y}$ and the corresponding reliability metric, Bhattacharyya parameter, can be expressed as
\begin{equation}
 Z( W ) = \sum\limits_{y \in \mathcal{Y}} {\sqrt {W( y|0)W(y|1)}}=Z_0.
\end{equation}
Applying channel polarization transform for $N=2^n$ independent uses of B-DMC $W$, after channel combining and splitting operation \cite{Polarcode_Arikan}, we can obtain a group of polarized channels $W_N^{(i)}: \mathcal{X}\to \mathcal{Y} \times \mathcal{X}^{i-1}$, $i=1,2,\cdots,N$. The Bhattacharyya parameters of these channels satisfy the following recursion
\begin{equation}\label{equation2}
\left\{\begin{array}{l}
Z\left( {W_N^{\left( {2i - 1} \right)}} \right) \le 2Z\left( {W_{N/2}^{\left( i \right)}} \right) - Z{\left( {W_{N/2}^{\left( i \right)}} \right)^2}\\
Z\left( {W_N^{\left( {2i} \right)}} \right) = Z{\left( {W_{N/2}^{\left( i \right)}} \right)^2}.
\end{array}\right.\
\end{equation}
By using of the channel polarization, the polar coding can be described as follows.

Given the code length $N$, the information length $K$ and code rate $R=K/N$, the indices set of polarized channels can be divided into two subsets: one set $\mathcal{A}$ to carry information bits and the other complement set $\mathcal{A}^c$ to assign the fixed binary sequence, named frozen bits. So a message block of $K=|\mathcal{A}|$ bits is transmitted over the $K$ most reliable channels $W_N^{(i)}$ with indices $i\in \mathcal{A}$ and the others are used to transmit the frozen bits.  So a binary source block $u_1^N$ consisting of $K$ information bits and $N-K$ frozen bits can be encoded into a codeword $x_1^N$ by
\begin{equation}\label{equation3}
x_1^N=u_1^N{\bf{G}}_N,
\end{equation}
where the matrix ${{\bf{G}}_N}$ is the $N$-dimension generator matrix. This matrix can be recursively defined as ${{\bf{G}}_N} = {{\bf{B}}_N}{\bf{F}}_2^{ \otimes n}$, where ``$^{\otimes n}$" denotes the $n$-th Kronecker product, ${\bf{B}}_N$ is the bit-reversal permutation matrix, and ${{\bf{F}}_2} = \left[ { \begin{smallmatrix} 1 & 0 \\ 1 &  1 \end{smallmatrix} } \right]$ is the $2\times2$ kernel matrix.

In this paper, we mainly use the trellis or factor graph based on the coding relationship $x_1^N=u_1^N{\bf{G}}_N$ to describe the structure of polar or RCPP codes, where the source bits are arranged by the bit-reversal order and the code bits by the natural order. On the other hand, we also introduce the dual trellis to simplify the analysis of puncturing schemes.
\begin{definition}
The dual trellis or dual factor graph is defined as a trellis deduced from the constraint $u_1^N=x_1^N\mathbf{G}_N^{-1}=x_1^N {\mathbf{B}_N} {\mathbf{F}_2^{\otimes n}}=x_1^N\mathbf{G}_N$ (see \cite[Lemma1]{LP_Goela}). Compared with the original trellis, in this dual trellis, the source bits are assigned by the natural order and the code bits by the bit-reversal order.
\end{definition}

For the construction of polar codes, the calculation of channel reliabilities and selection of good channels are the critical steps. In this paper, for the convenience of theoretic analysis, we mainly use Bhattacharyya parameter to indicate the channel reliability.

As pointed in \cite{Polarcode_Arikan}, polar codes can be decoded by the SC decoding algorithm with a low complexity $O(N\log N)$. Furthermore, many improved SC decoding algorithms, such as, successive cancellation list (SCL) \cite{SCL_Tal,SCL_Add}, successive cancellation stack (SCS) \cite{SCS_Niu}, successive cancellation hybrid (SCH) \cite{SCH_Chen}, and CRC aided (CA)-SCL/SCS \cite{SCL_Tal,CASCL_Niu,ASCL_Li} decoding can be applied to improve the performance of polar codes.
\subsection{Upper Bounds of Bhattacharyya Parameters for Polar Codes}
The channel index $i$ can be expanded as
\begin{equation}\label{binary_expansion}
i = 1 + \sum\limits_{l = 1}^n {b_l{2^{n-l}}},
\end{equation}
where $\left(b_1,\cdots,b_l,\cdots,b_n\right)$ denote the $n$-bit binary expansion of $i-1$ and $l$ is the polarization level.

Let $A_0=a_0=\log_2(Z_0)$ and $A_l=\log_2(Z_l^{u})$, where $Z_l^{u}$ denotes the upper bound of the Bhattacharyya parameter at level $l$. According to the analytical idea of asymptotic convergence of Bhattacharyya parameter in \cite{polar_rate}, the iteration of the upper bound in the logarithmic domain can be expressed as
\begin{equation}\label{log_BP_bound}
\left\{ \begin{aligned}
A_l &= A_{l-1}+1,  & {\text{   if    }}  b_{l}=0,\\
A_l &= 2A_{l-1},    &{\text{    if    }} b_{l} = 1.
\end{aligned}\right.
\end{equation}

\section{RCPP Codes}
\label{section_III}
In this section, we introduce the definition of RCPP codes and describe the corresponding coding and decoding process. Then, we review the puncturing modes of RCPP codes, such as the C0 and C1 modes. In the end, we analyze the upper bound of Bhattacharyya parameter for both modes.
\subsection{Definition of RCPP Codes}
RCPP codes are a kind of rate- and length- compatible polar codes.
The entire encoding process can be described by two steps. In the first step, an original $K$-bit information block is coded by the coding constraint (\ref{equation3}), that is, a binary source block $u_1^N$ consisting of the information subvector $u_\mathcal{A}$ and the frozen subvector $u_{\mathcal{A}^c}$ is encoded into a code block $x_1^N$.

Then in the second step, in order to adapt the rate variation, the length of $N$-bit code block is shortened according to the puncturing table. Here, the puncturing table $\mathscr{T}_N$ is defined as
\begin{equation}
\label{equ1}
	\mathscr{T}_N=\left(t_1,t_2,\cdots,t_N \right)
\end{equation}
with $t_i\in \{0,1\}$, $i\in \mathcal{I}$, where $t_i=0$ means that the code bit $x_i$ in the corresponding index is not to be transmitted, and vise versa.

Let $\mathcal{B}=\left\{i|t_i=0\right\}$ and $\mathcal{B}^c=\left\{i|t_i=1\right\}$ denote the puncturing set and the complement set, whereby the corresponding cardinalities are defined as $\left|\mathcal{B}\right|=Q=N-M$ and $\left|\mathcal{B}^c\right|=M$ respectively. Since the code length of the original polar codes is limited to a power of $2$, without loss of generality, it suffices to concern only on the case that $2^{n-1}< M\leq2^{n}$. So after puncturing $Q$ code bits, an $M$ length RCPP codeword $x_{\mathcal{B}^c}$ can be obtained from the table $\mathscr{T}_N$. Accordingly, the code rate of RCPP coding scheme can be defined as $R=K/M$.

The encoding process of RCPP code can also be equally described based on the puncturing of source bits. Firstly, we introduce a puncturing set of source bits $\mathcal{D}$, which satisfies $\left|\mathcal{D}\right|=\left|\mathcal{B}\right|=Q$. After deleting $Q$ bits, a minished source vector $u_{\mathcal{D}^c}$ is obtained. Then the codeword $x_{\mathcal{B}^c}$ with the code length $M$ can be written by
\begin{equation}
x_{\mathcal{B}^c}=u_{\mathcal{D}^c}\mathbf{G}_M,
\end{equation}
%\vspace{-0.2em}
where the dimension-reduction generator matrix $\mathbf{G}_M$ is obtained by eliminating the columns corresponding to the set $\mathcal{B}$ and rows corresponding to the set $\mathcal{D}$ from the matrix $\mathbf{G}_N$.

The construction of RCPP codes is similar to that of polar codes. Let $\mathbb{W}$ denote the punctured channel. For RCPP codes, the transmission channel $\widetilde{W}$ can be regarded as a compound of the original B-DMC and the punctured channel, i.e., $\widetilde{W}=\left\{W,\mathbb{W}\right\}$. Similarly, under the puncturing operation, after channel splitting and combing, we can also obtain a group of polarized channels $\left\{\widetilde{W}_N^{(i)}\right\}$ and the corresponding Bhattacharyya parameters $Z \left (\widetilde{W}_N^{(i)}\right)$.

The decoder of RCPP codes has the same decoding structure as that of polar codes. However, the initialization of bit LLRs for those punctured bits is different for the C0 or C1 modes, which will be further explained in the next subsection.

\subsection{Puncturing Modes of RCPP Codes}
By now, we have two puncturing modes for RCPP codes: C0 mode and C1 mode. For the former, the code bits in the puncturing set $\mathcal{B}$ are deleted in the encoder and their values are unknown in the the decoder. Thus, the transition probabilities of punctured channel $\mathbb{W}$ are $\mathbb{W}\left( {y_i}\left| 0 \right. \right)=\mathbb{W}\left( {y_i}\left| 1 \right. \right)=\frac{1}{2}$. The corresponding channel capacity is zero, that is, $I(\mathbb{W})=0$. Given a punctured bit $x_i$ ($i\in \mathcal{B}$), for the C0 mode, the corresponding LLR in the decoder is $\mathbb{L}(y_i)=\ln \frac{\mathbb{W}\left( {y_i}\left| 0 \right. \right)}{\mathbb{W}\left( {y_i}\left| 1 \right. \right)} = 0$.

On the contrary, for the latter, the punctured bits are set as frozen bits \cite{Novel_Punc} in the encoder and their values are fixed and known in the decoder. Suppose the fixed value is zero, that is, $y_i=0$, the transition probabilities are $\mathbb{W}\left( 0\left| 0 \right. \right)=1$, $\mathbb{W}\left( 0\left| 1 \right. \right)=0$. So the channel capacity is one, that is, $I(\mathbb{W})=1$ and the corresponding LLR is $\mathbb{L}(y_i)= +\infty$.

For two channel polarization, the factor graphs of polar codes under the C0 and C1 modes are shown in Fig. \ref{fig_two_channel}.
\begin{figure}[h]
  \centering
  \includegraphics[width=1\columnwidth]{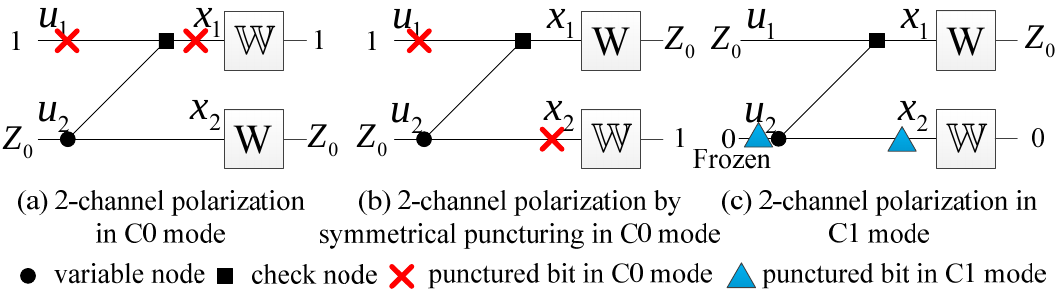}
  \caption[c]{Factor graph representation for two channel polarization under two puncturing modes}
  \label{fig_two_channel}
\end{figure}

In this figure and throughout the paper, black circle and black squares denote the variable and check nodes of the factor graph, meanwhile, red cross and blue triangle stand for the punctured bits under the C0 or C1 modes respectively. In Fig. \ref{fig_two_channel}(a) (\ref{fig_two_channel}(b)), the code bit $x_1$ ($x_2$) is punctured and the corresponding Bhattacharyya parameter is $Z(\mathbb{W})=1$. The relative puncturing tables are $\mathscr{T}_2=(0,1)$ and $\mathscr{T}_2=(1,0)$ respectively. On the other hand, in Fig. \ref{fig_two_channel}(c), the code bit $x_2$ are punctured and the corresponding Bhattacharyya parameter is $Z(\mathbb{W})=0$, meanwhile, the source bit $u_2$ is a frozen bit. In this scheme, the puncturing table is $\mathscr{T}_2=(1,0)$.

\begin{lemma}\label{lemma2}
For two channel polarization under the C0 mode, both puncturing tables $\mathscr{T}_2=(0,1)$ or $\mathscr{T}_2=(1,0)$ can generate the same polarization results.
\end{lemma}

\begin{lemma}\label{lemma3}
For two channel polarization under the C1 mode, the puncturing table satisfies $\mathscr{T}_2=(1,0)$ and the source bit $u_2$ should be assigned a fixed value. In addition, the reliabilities of polarized channels satisfies $Z\left(\widetilde{W}_2^{(1)}\right)=Z(W)$ and $Z\left(\widetilde{W}_2^{(2)}\right)=0$ respectively.
\end{lemma}

These two lemmas will be proved in the Appendix. Furthermore, they can be recursively applied in the process of $N$ channels polarization.

\begin{lemma}\label{C1_mode_BC}
 Given a RCPP code constructed under the C1 mode, the Bhattacharyya parameters of the polarized channels are smaller than those of the original polarized channels, that is, $Z\left(\widetilde {W}_N^{(i)}\right)<Z\left({W}_N^{(i)}\right)$.
\end{lemma}

This lemma reveals that the puncturing under the C1 mode will improve the reliability of each polarized channel and will be proved in the Appendix.

\subsection{Upper bounds of Bhattacharyya Parameters for RCPP Codes}
Let $\widetilde{Z}_l^u$ denote the upper bounds of the Bhattacharyya parameters under puncturing and $\mathscr{A}_n=\log_2\left(\widetilde{Z}_l^{u}\right)$, on the corresponding trellis, the upper bounds of Bhattacharyya parameters can be iteratively evaluated by considering the reliability difference of the polarized channels.

In the first case, a pair of independent polarized channels ${\widetilde{W}}_{N/2}^{(i)}$ with the same reliability are considered. By using the same binary expansion in (\ref{binary_expansion}),  we can write the iteration of these bounds in the logarithmic domain as
\begin{equation}\label{Bhattaharyya_bound_log}
\left\{ \begin{aligned}
\mathscr{A}_l &= \mathscr{A}_{l-1}+1,  & {\text{   if    }}  b_{l}=0,\\
\mathscr{A}_l &= 2\mathscr{A}_{l-1},    &{\text{    if    }} b_{l} = 1.
\end{aligned}\right.
\end{equation}

In the second case, we consider a pair of channels have different reliabilities, meanwhile, one is a punctured channel and the other is a polarized channel.

For the C0 mode, from Lemma \ref{lemma2}, the upper bounds at level $l$ can be iteratively calculated in the logarithmic domain as
\begin{equation}\label{bound_log_C0}
\left\{ \begin{aligned}
\mathscr{A}_l &= 0,                             & {\text{   if    }}  b_{l}=0,\\
\mathscr{A}_l &= \mathscr{A}_{l-1},     &{\text{    if    }} b_{l} = 1.
\end{aligned}\right.
\end{equation}

Accordingly, for the C1 mode, from Lemma \ref{lemma3}, the upper bounds can also be written by
\begin{equation}\label{bound_log_C1}
\left\{ \begin{aligned}
\mathscr{A}_l &= \mathscr{A}_{l-1},   & {\text{   if    }}  b_{l}=0,\\
\mathscr{A}_l &=-\infty,                        &{\text{    if    }} b_{l} = 1.
\end{aligned}\right.
\end{equation}

Now, we consider the third case, that is, a pair of channels have different reliabilities and the Bhattacharyya parameters of both channels are not equal to $0$ or $1$. Obviously, the channels in this case are obtained from channel polarization in the first and second cases. Let ${\mathscr{A}_l}'$ and ${\mathscr{A}_l}''$ denote the upper bounds of these two channels in a log-scale respectively.

\begin{lemma}\label{C0_mode_UB_set}
For the C0 mode, the upper bound of the Bhattacharyya parameter at level $l-1$ should be set to the maximum value, that is, $\mathscr{A}_{l-1}=\max\left\{{\mathscr{A}_{l-1}}',{\mathscr{A}_{l-1}}''\right\}$. Furthermore, the upper bounds at level $l$ can be iteratively calculated by (\ref{Bhattaharyya_bound_log}).
\end{lemma}
\begin{IEEEproof}
In this case, one channel is obtained from the $M$ B-DMCs polarization and the other from the polarization of punctured channels and B-DMCs. We select the maximum value of as the upper bound of Bhattacharyya parameter at level $l-1$, which can indicate the worse reliability of RCPP codes under the C0 mode.
\end{IEEEproof}

\begin{lemma}\label{C1_mode_UB_set}
For the C1 mode, the upper bound at level $l-1$ should be set to the minimum value, that is, $\mathscr{A}_{l-1}=\min\left\{{\mathscr{A}_{l-1}}',{\mathscr{A}_{l-1}}''\right\}$.
\end{lemma}
\begin{IEEEproof}
By Lemma \ref{C1_mode_BC}, if the minimum value is selected as the upper bound of Bhattacharyya parameter at level $l-1$, this bound ensures that the reliability of each polarized channel under the C1 mode is better than that of the original polarized channel.
\end{IEEEproof}

So for the C0 mode, the upper bounds of Bhattacharyya parameters can be iteratively calculated by using (\ref{Bhattaharyya_bound_log}) and (\ref{bound_log_C0}) in logarithmic domain. Correspondingly, for the C1 mode, the upper bounds can be evaluated by using (\ref{Bhattaharyya_bound_log}) and (\ref{bound_log_C1}).

\begin{figure}[h]
  \centering
  \includegraphics[width=1\columnwidth]{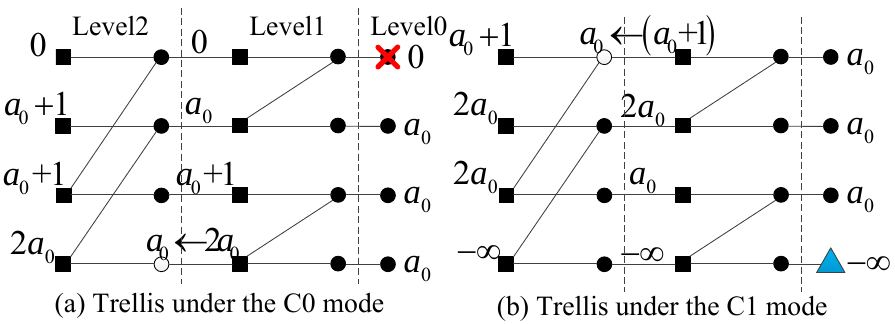}
  \caption[c]{Puncturing examples under the C0 or C1 modes}
  \label{fig_trellis}
\end{figure}

\begin{example}\label{example1}
Figure \ref{fig_trellis} gives two examples of $N=4$ channel polarization with $Q=1$ punctured bits under two puncturing modes. In this figure, the numbers next to the variable nodes of each level denote the upper bounds (in the logarithmic domain) of Bhattacharyya parameters of the transmission channels.

In Fig. \ref{fig_trellis}(a), under the C0 mode, the first code bit is punctured and the puncturing table is $\mathscr{T}_4=\{0,1,1,1\}$. Applying Lemma \ref{C0_mode_UB_set}, the number next to the white circle node should be changed from $2a_0$ to $a_0$. On the other hand, in Fig. \ref{fig_trellis}(b), under the C1 mode, the fourth code bit are punctured and the puncturing table is $\mathscr{T}_4=\{1,1,1,0\}$. From Lemma \ref{lemma3}, we can conclude that source bit $u_4$ at level $2$ should be set to a fixed value. Furthermore, by Lemma \ref{C1_mode_UB_set}, the number next to the white circle node should be altered from $a_0+1$ to $a_0$.
\end{example}
\section{Code-Tree Characterization}\label{section_IV}
In this section, we show how to use a code tree to describe the process of channel polarization with puncturing operation. Based on the tree structure, we introduce the concepts of polar spectra (PS) and path weight enumeration function (PWEF). Furthermore, three types of spectrum distance, such as SD0, SD1, and JSD, are introduced as key performance metrics to indicate the distribution of polar spectra.
\subsection{Code Tree}
Code tree is a compact representation of trellises for polar or RCPP codes. Given the parent code length $N = 2^n$, the code tree $\mathcal{T}=(\mathcal{V}, \mathcal{P})$ is a binary tree, where $\mathcal{V}$ and $\mathcal{P}$ denote the set of nodes and the set of edges or branches, respectively.

The depth of a node is the length of the path from the root to this node. The set of all the nodes at a given depth $l$ is denoted by $\mathcal{V}_l(l = 0, 1, 2, \cdots, n)$. The root node has a depth of zero. The nodes in the set $\mathcal{V}_l$ can be enumerated one-by-one from left to right on the tree, that is, $v_{l,m}\left(m = 1, 2, \cdots, 2^l\right)$ denotes the $m$-th node in $\mathcal{V}_l$. Except for the nodes at the $n$-th depth, each $v_{l,m} \in \mathcal{V}_l$ has two descendants in $\mathcal{V}_{l+1}$, and the two corresponding branches are labeled as $0$ and $1$, respectively. The nodes $v_{n,m}\in \mathcal{V}_n$ are called leaf nodes. Let $\mathcal{T}(v_{l,m})$ denote a subtree with a root node $v_{l,m}$. The depth of this subtree can be defined as the difference between the depth of leaf node and that of the root node, that is, $n-l$.

Recall that channel index $i$ can be expanded by a binary sequence (\ref{binary_expansion}), hence, we can use this sequence to label a path $\omega_n^{(i)}=\left(b_1,\cdots,b_l,\cdots,b_n\right)$ from the root node to one leaf node\footnote{Throughout this paper, we use $\omega_n^{(i)}$ to denote a path on the code tree with a depth $n$. The superscript will be stripped without causing confusion.}, whereby one branch between depth $l-1$ and depth $l$ is assigned a bit value $b_l$. Let $\omega_l=\left(b_1,\cdots,b_l\right)$ denote a partial path from the root node to a node in depth $l$. Note that, using this labeling method\footnote{Hereafter, in all following examples, we will use the same labeling to enumerate the nodes or branches on the code tree.}, the source bits corresponding to the leaf nodes are arranged by a natural order.

Considering the one-to-one correspondence between the channel $i$ and the path $\omega_n$, we use the reliability of channel $i$ to denote the reliability of the corresponding path $\omega_n$. Furthermore, given an end node $v_{n,m}$ of one path, the reliability of this node can also be evaluated by that of the path and denoted by $B\left(v_{n,m}\right)$.

Figure \ref{fig_code_tree} shows two examples of code tree with a parent code length $N=4$ for the C0 and C1 modes. Each code tree is a compact presentation of the trellis in Example \ref{example1}. Each depth in the tree corresponds to one level on the trellis shown in Fig. \ref{fig_trellis}. The source bit $u_3$ corresponds to a binary expansion $(1,0)$ and this sequence is assigned to a path $\omega_2^{(3)} =(1,0)$ which is also indicated by a node sequence $(v_{0,1},v_{1,2},v_{2,3})$. In Example \ref{example1}, the punctured source bits under the C0 or C1 modes are $u_1$ or $u_4$ respectively.

As shown in Fig. \ref{fig_code_tree}(a), the leftmost path is pruned and there are two subtrees, such as $\mathcal{T}\left(v_{2,2}\right)$ and $\mathcal{T}\left(v_{1,2}\right)$. The reliability metrics corresponding to the root nodes of these subtrees are $a_0+1$ and $a_0$, respectively. Similarly, in Fig. \ref{fig_code_tree}(b), the rightmost path is pruned and there are two subtrees, such as $\mathcal{T}\left(v_{1,1}\right)$ and $\mathcal{T}\left(v_{2,3}\right)$. The reliability metrics corresponding to the root nodes of these subtrees are $a_0$ and $2a_0$, respectively.

\begin{figure}[h]
  \centering
  \includegraphics[width=1\columnwidth]{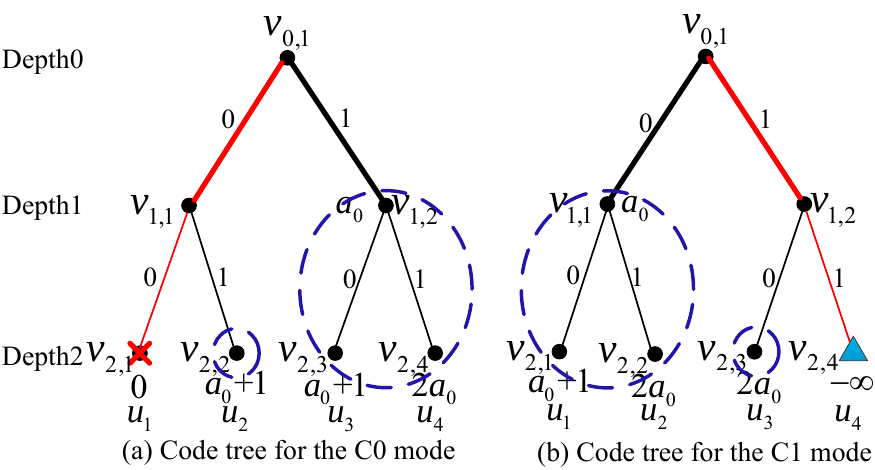}
  \caption[c]{Code tree example for the C0 and C1 modes}
  \label{fig_code_tree}
\end{figure}

Generally, for the punctured source bits on the code tree, we have the following lemmas.
\begin{lemma}\label{lemma5}
For the C0 mode, if one punctured leaf node is a right descendent of a subtree, then the left descendent of this subtree is also punctured.
\end{lemma}
\begin{IEEEproof}
Given an arbitrary subtree $\mathcal{T}\left(v_{{n-1},m}\right)$, it has a left descendent $v_{n,2m-1}$ and a right one $v_{n,2m}$. These two leaf nodes are corresponding to two channel polarization. Under the C0 mode, if this subtree is punctured one bit, by Lemma \ref{lemma2}, the left descendent $v_{n,2m-1}$ should be firstly punctured. Furthermore, if the right descendent $v_{n,2m}$ is punctured, this subtree will be fully deleted.
\end{IEEEproof}

\begin{lemma}\label{lemma6}
For the C1 mode, if one punctured leaf node is a left descendent of a subtree, then the right descendent of this subtree is also punctured.
\end{lemma}
\begin{IEEEproof}
According to Lemma \ref{lemma3}, using the same argument as that of Lemma \ref{lemma5}, we can obtain the conclusion.
\end{IEEEproof}

\subsection{Polar Spectra and Path Weight Enumeration Function}
Now, we consider the performance analysis of polar or RCPP codes on the code tree. Because each path $\omega_n$ is relative to a polarized channel $i$ (refer to footnote 1), the corresponding upper bound of Bhattacharyya parameter $A_n\left(\omega_n\right)=\log_2\left(Z_n^{u}\left(\omega_n\right)\right)$ can be iteratively evaluated by using (\ref{log_BP_bound}) or (\ref{Bhattaharyya_bound_log}) in a log scale. This calculation includes two operators, that is, the adding-one operator $O_a: \mathbb{R} \to \mathbb{R}, O_a(x) = x + 1$ and the doubling operator $O_d: \mathbb{R} \to \mathbb{R}, O_d(x) = 2x$. Given the initial value of the root node $A_0=a_0$, a sequence of numbers $A_1(\omega_1),\cdots,A_n(\omega_n)$ can be recursively calculated by
\begin{equation}
{A_{l }}\left( {{\omega_{l}}} \right) = {g_l}\left( {{A_{l-1}}\left( {{\omega_{l-1}}} \right)} \right)
\end{equation}
where $g_l\in \left\{O_a,O_d\right\}$. From (\ref{Bhattaharyya_bound_log}), when $b_l=0$, the operator $g_l=O_a$ is applied and when $b_l=1$, the operator $g_l=O_d$ is used.

\begin{definition}
The path weight $d_H(\omega_n)$ is defined as the Hamming weight of the binary vector $\omega_n=\left(b_1,\cdots,b_l,\cdots,b_n\right)$, that is, $d_H(\omega_n)=k=\left|\left\{l:b_l=1\right\}\right|$. Furthermore, we can define the complemental path weight $f_H(\omega_n)$ as the complemental Hamming weight of the path $\omega_n$, that is, $f_H(\omega_n)=r=n-k=\left|\left\{l:b_l=0\right\}\right|$.
\end{definition}

Obviously, during the $n$ iterations of the $A_n(\omega_n)$, we enumerate doubling $d_H(\omega_n)$ times and adding-one $f_H(\omega_n)$ times. And these two weights satisfy $d_H(\omega_n)+f_H(\omega_n)=n$.

\begin{theorem}\label{theorem1}
Given a path $\omega_n$ with the path weight $d_H(\omega_n)$ and the complemental path weight $f_H(\omega_n)$, the corresponding upper bound of Bhattacharyya parameter $A_n\left(\omega_n\right)$ can be further bounded by
\begin{equation}\label{Bhattaharyya_loguplow}
{A_n^l}( {\omega_n} ) \le {A_n}( {\omega_n} ) \le {A_n^u}( {\omega_n} ),
\end{equation}
where ${A_n^l}( {\omega_n} )$ and ${A_n^u}( {\omega_n} )$ satisfy
\begin{equation}
\left\{ \begin{aligned}
{A_n^l}( {\omega_n} )&={2^{{d_H}\left( {{\omega_n}} \right)}}{a_0} + {f_H}\left( {{\omega_n}} \right),\\
{A_n^u}( {\omega_n} )&={2^{{d_H}\left( {{\omega_n}} \right)}}\left( {{a_0} + {f_H}\left( {{\omega_n}} \right)} \right).
\end{aligned}\right.
\end{equation}
\end{theorem}
\begin{IEEEproof}
Suppose the path $\omega_n=\left(b_1,\cdots,b_l,\cdots,b_n\right)$ is relative to an operator sequence $\left\{g_l\right\}_{l=1}^n$ and the beginning of the sequence is $g_1=O_a$ (If the beginning is $g_1=O_d$, we will check the sequence and find a partial one with the beginning of $O_a$). So there exists $l\in\left\{2,\cdots,n\right\}$ for which $g_{l-1}=O_a$ and $g_l=O_d$. According to the argument in \cite[Lemma 1]{polar_rate}, swapping $g_{l-1}$ and $g_{l}$ will decrease the result of recursion. So after continuously swapping over the sequence, we can obtain a lower bound on $A_n(\omega_n)$ corresponds to choosing $g_1=\cdots=g_k=O_d$ and $g_{k+1}=\cdots=g_{n}=O_a$, that is, $A_n(\omega_n)\ge O_a^{n-k}\left(O_d^{k}(a_0)\right)={A_n^l}( {\omega_n} )$. By a similar argument, we can find an upper bound $A_n^u(\omega_n)$ which is proved in \cite{polar_rate}.
\end{IEEEproof}

Since path weight and complemental path weight indicate the reliability of polarized channel, given the end node $v_{n,m}$ of the path $\omega_n$, we can use the lower bound ${A_n^l}( {\omega_n} )$ to present the reliability of this node, that is
\begin{equation}\label{reliability_approx_node}
B\left(v_{n,m}\right)={2^{{d_H}\left( {{\omega_n}} \right)}}{a_0} + {f_H}\left( {{\omega_n}} \right).
\end{equation}

\begin{definition}
Polar spectra are defined by the distribution of path weight or complemental path weight on the code tree and characterized by two sets, PS1 and PS0, to count the number of paths with a certain path weight or complemental path weight respectively. Let $\left\{H_M^{(k)}, 0\le k\le n\right\}$ denote the PS1 set of a RCPP code with the code length $M$, where $H_M^{(k)}$ represents the number of paths with a path weight $k$ on the code tree after the puncturing. Similarly, $\left\{C_M^{(r)}, 0\le r\le n\right\}$ denotes the corresponding PS0 set, where $C_M^{(r)}$ means the number of paths with a complemental path weight $r=n-k$.
\end{definition}

\begin{remark} For the original polar code, i.e., $M=N$, due to the structure of perfect binary tree, the elements in PS1 and PS0 satisfy $H_N^{(k)}=\binom{n}{k}$ and $C_N^{(r)}=\binom{n}{r}$ respectively. By Theorem \ref{theorem1}, the set of PS1 or PS0 directly determine the reliability of polarized channel. Hence, how to approach these original polar spectra is the aim of optimal puncturing for RCPP codes.
\end{remark}

\begin{definition}
Path weight enumeration function (PWEF) is characterized by two types of polynomials whose coefficients are taken from the corresponding polar spectra. Given the PS1 set $\left\{H_M^{(k)},0\le k\le n\right\}$, PWEF on the path weight (PWEF1) can be defined as
${\mathcal H}\left( X \right)=\sum\limits_{k = 0}^n {H_M^{\left( k \right)}{X^k}}$, where $X$ is a dumb variable.
Similarly, given the PS0 set $\left\{C_M^{(r)},0\le r\le n\right\}$, PWEF on the complemental path weight (PWEF0) can be defined as
${\mathcal C}\left( X \right)=\sum\limits_{r = 0}^n {C_M^{\left( r \right)}{X^r}}$.
\end{definition}

\begin{lemma}\label{lemma7}
The PS1 and PS0 of an original polar code are symmetric, that is, given a pair of paths $\omega_n^{(i)}$ and $\omega_n^{(N+1-i)}$, we have $f_H\left(\omega_n^{(i)}\right)=d_H\left(\omega_n^{(N+1-i)}\right)$ and $d_H\left(\omega_n^{(i)}\right)=f_H\left(\omega_n^{(N+1-i)}\right)$. Further, the corresponding PWEF1 and PWEF0 satisfy
${\mathcal H}\left( X \right)={\mathcal C}\left( X \right)=\sum\limits_{k = 0}^n \binom{n}{k}{X^k}$.
\end{lemma}

\subsection{Spectrum Distance}
We introduce two types of spectrum distances, SD1 and SD0, defined by the expectation of path weight and complemental weight respectively.
\begin{definition}
The spectrum distance for path weight (SD1) is given by
\begin{equation}
\begin{aligned}
d_{avg}&=\mathbb{E}\left[d_H(\omega_n)\right]=\frac{1}{M}\left.\frac{d\mathcal{H}(X)}{dX}\right|_{X=1}\\
           &=\sum\limits_{k = 0}^n {P_1(n,k,Q)k}= \sum\limits_{k = 0}^n {\frac{H_M^{\left( k \right)}}{M}k}
\end{aligned}
\end{equation}
where $P_1(n,k,Q)=\frac{H_M^{\left( k \right)}}{M}$ is the probability of  path weight $k$ for a RCPP code with $Q$ bits puncturing.
Correspondingly, the spectrum distance for complemental path weight (SD0) is given by
\begin{equation}
\begin{aligned}
\lambda_{avg}&=\mathbb{E}\left[f_H(\omega_n)\right]=\frac{1}{M}\left.\frac{d\mathcal{C}(X)}{dX}\right|_{X=1}\\
           &=\sum\limits_{r = 0}^n {P_0(n,r,Q)r}= \sum\limits_{r = 0}^n {\frac{C_M^{\left(r \right)}}{M}r}
\end{aligned}
\end{equation}
where $P_0(n,r,Q)=\frac{C_M^{\left( r \right)}}{M}$.
\end{definition}

\begin{definition}
In addition, we can define the joint spectrum distance (JSD) as follows
\begin{equation}
\begin{aligned}
d_{avg}+\lambda_{avg}&=\mathbb{E}\left[d_H(\omega_n)\right]+\mathbb{E}\left[f_H(\omega_n)\right]\\
                        &=\sum\limits_{k = 0}^n {P_1(n,k,Q)k}+\sum\limits_{r = 0}^n {P_0(n,r,Q)r}.
\end{aligned}
\end{equation}
\end{definition}
Hereafter, we use SD0/SD1/JSD as the main metrics to evaluate and optimize the puncturing table. If these metrics of one puncturing scheme are very close to those of the original polar code, this scheme will generate an optimal RCPP code.
\begin{corollary}\label{corollary2}
The SD1 and SD0 of the original polar code are $\mathbb{E}\left[d_H\left(\omega_n\right)\right]=\frac{n}{2}$ and $\mathbb{E}\left[f_H\left(\omega_n\right)\right]=\frac{n}{2}$ respectively.
\end{corollary}
\begin{IEEEproof}
The proof is direct. By Lemma \ref{lemma7}, since the probability $P_1(n,k,0)$ of the original polar code obeys the binomial distribution, we can write
\begin{equation}
\begin{aligned}
\mathbb{E}\left[d_H\left(\omega_n\right)\right]=\sum\limits_{k = 0}^n {P_1(n,k,0)k}
                                                                 =\sum\limits_{k = 0}^n{\binom{n}{k}\frac{k}{2^n}}=\frac{n}{2}.
\end{aligned}
\end{equation}
The derivation of SD0 is similar and omitted.
\end{IEEEproof}

\section{Optimal Puncturing for the Capacity-Zero Mode}
\label{section_V}
In this section, the quasi-uniform puncturing (QUP) algorithm is described and proved to maximize the spectrum distances SD1 and JSD.
\subsection{Single Bit Puncturing}
Consider the single-bit puncturing under C0 mode. In this case, puncturing at different locations is equivalent regardless of a slight variation in its code tree.
\begin{theorem}\label{theorem3}
For the C0 mode, when only one bit is punctured, puncturing any code bit $x_i$ is equivalent to puncturing the first code bit $x_1$, meanwhile, the source bit $u_1$ is punctured.
\end{theorem}
\begin{proof}
The polarized transformation $u_1^N\mathbf{G}_N=x_1^N$ is inverted as $u_1^N=x_1^N\mathbf{G}_N$ \cite{LP_Goela}. Hence, the first source bit can be written as $u_1=\sum \limits_{i = 1}^N x_i$. This bit $u_1$ is constrained by all the code bits via the modulo-2 operation. When any one code bit is punctured under the C0 mode, it is easy to see that the source $u_1$ bit must be punctured according to Lemma \ref{lemma2}.
\end{proof}

\begin{theorem}\label{theorem4}
For single bit puncturing under the C0 mode, the maximal path weight on the code tree is $n-1$, similarly, the maximal complemental path weight is also $n-1$.
\end{theorem}
\begin{IEEEproof}
According to Theorem \ref{theorem3}, after single bit puncturing under the C0 mode, the leftmost path $\omega_n^{(1)}=(0,0,\cdots,0)$ ($b_l=0,l=1,2,\dots,n$) corresponding to the first source bit $u_1$ on the code tree is pruned. The nodes associated to this pruned path are $v_{l,1}, l=0,1,\cdots,n$. So the original code tree is decomposed into $n$ subtrees and each has a root node $v_{l,2}, l=1,2,\cdots,n$.

Since every bit $b_l$ is punctured, by using (\ref{bound_log_C0}), we can calculate the reliability metric of each root node, that is, $B\left(v_{l-1,1}\right) \rightarrow B\left(v_{l,2}\right)$. For one root node $v_{l,2}$, we regard its reliability metric as an inheritance from the predecessor $\left(v_{l-1,1}\right)$. Applying Lemma \ref{C0_mode_UB_set} and (\ref{Bhattaharyya_bound_log}) on all the subtrees, we can evaluate the reliability of each leaf node. Obviously, the rightmost subtree $\mathcal{T}\left(v_{1,2}\right)$ has the largest depth and its rightmost path has the maximal path weight $n-1$. Furthermore, the leftmost path of each subtree has the maximal complemental path weight $n-1$.
\end{IEEEproof}

For a code tree shown in Fig. \ref{fig_code_tree}(a), the rightmost path on the subtree $\mathcal{T}\left(v_{1,2}\right)$ has the maximal path weight $1$. And the leftmost path on all subtrees has the maximal complemental path weight $1$. The reliability metrics of two root nodes $v_{1,2}$ and $v_{2,2}$ are inherited from those of the predecessors $v_{0,1}$ and $v_{1,1}$ respectively, that is $B(v_{0,1})=a_0 \rightarrow B(v_{1,2})$ and $B(v_{1,1})=(a_0+1) \rightarrow B(v_{2,2})$.

\subsection{Quasi-Uniform Puncturing Algorithm}
\label{section_puncturing}
The quasi-uniform puncturing (QUP) algorithm proposed in \cite{RCPP_Niu} can be outlined as follows.
\medskip
\begin{enumerate}[ \textbf{Stage} 1)]
\item
Initialize the table $\mathscr{T}_N$ as all ones, and then set the first $Q$ bits as zeros;
\item
Perform bit-reversal permutation on the table $\mathscr{T}_N$ and obtain the puncturing table.
\end{enumerate}
\medskip

The puncturing table generated by the QUP algorithm is constructive and regular, thereby providing useful tool for the practical application of coding and decoding.

\begin{example}
$N=8, M=5, Q=3$. The initial table is $\mathscr{T}_8=\left(00011111 \right)$. After bit-reversal permutation, the puncturing table is $\mathscr{T}_8=\left( 01010111 \right)$, that means the code bits $x_1$, $x_3$, and $x_5$ should be punctured.
\end{example}

\begin{theorem}\label{theorem5}
The punctured positions in the QUP puncturing table are roughly uniform, that is, the distance between any two neighboring punctured positions, $D$, satisfies $2^{(n-L-1)} \le D \le 2^{(n-L)}$ where $L=\lfloor \log_2 {Q} \rfloor$.
\end{theorem}

This theorem is proved in \cite{RCPP_Niu}.

\begin{lemma}\label{lemma8}
For the QUP algorithm, the source bit vector $u_1^Q$ is punctured, that is, $\mathcal{D}=\left\{1,2,\cdots,Q\right\}$. Equally, the $Q$ leftmost leaf nodes on the code tree are pruned.
\end{lemma}
\begin{IEEEproof}
From the operation of QUP algorithm, the puncturing set of code bits $\mathcal{B}$ is generated after bit-reversal permutation. Recall that the source vector can be written by $u_1^N=x_1^N {\mathbf{B}_N}{\mathbf{F}_2^{\otimes n}}$, so the source vector $u_1^N$ is punctured by the natural order, that is, $\mathcal{D}=\left\{1,2,\cdots,Q\right\}$. Equally, the leaf nodes set $\left\{v_{n,m}\left|m\in\mathcal{D}\right.\right\}$ on the code tree and the corresponding subtrees are pruned. Meanwhile, each of the rest subtrees has a different depth.
\end{IEEEproof}
\subsection{Optimal Puncturing Table}
Theoretically, the optimal puncturing table of RCPP codes can be optimized by a brute-force search of the distance spectra (for ML decoding) or BLER bounds (for SC decoding). However, the exhausted search for all the puncturing patterns is difficult to be realized. We are, therefore, concerned with the puncturing scheme to optimize the spectrum distances.

\begin{lemma}\label{lemma9}
For a subtree after any puncturing scheme under the C0 mode, suppose $v_{l,m}$ and $v_{l-1,\left\lceil \frac{m}{2} \right\rceil}$ are a root node and its predecessor respectively. Let $\zeta_n$ denote a pruned path from the original root $v_{0,1}$ to a punctured leaf node and containing the leftmost path of the subtree $\mathcal{T}\left(v_{l-1,\left\lceil \frac{m}{2} \right\rceil}\right)$. So the reliability of the root node $v_{l,m}$ can be addressed by
\begin{equation}
B\left(v_{l,m}\right)=a_0+f_H\left(\zeta_n\right)-1-(n-l)
\end{equation}
where $f_H\left(\zeta_n\right)$ is the complemental path weight.
\end{lemma}
\begin{IEEEproof}
Given the puncturing set of source bits $\mathcal{D}$ for an arbitrary puncturing, the leaf nodes pertaining to this set can be bit-by-bit punctured on the original code tree. Obviously, the source bit $u_1$ should be punctured firstly and the reliability metrics of the root nodes on the resulted subtrees can be inherited from the predecessors on the pruned path $\omega_n^{(1)}$ by Theorem \ref{theorem4}. Since the pruned path $\omega_n^{(1)}$ is an all-zero path, the corresponding path weight is $0$ and only the complemental path weight affects the calculation of the reliability. Assuming that one root node is $v_{l,2}$ and the corresponding partial path (from the original root to this node) is $\phi_l$, the reliability of this root can be expressed as
$B\left(v_{l,2}\right)=B\left(v_{l-1,1}\right)=a_0+f_H \left( \phi_l\right)$.

Furthermore, the rest punctured source bits can be pruned from these subtrees. Each subtree $\mathcal{T}\left(v_{l,2}\right)$ can be regarded as a perfect code tree. According to Lemma \ref{lemma5}, the leftmost path on these trees will be punctured and a group of new subtrees are generated. Then, for a root node $v_{l,m}$ on a final subtree, assuming the corresponding path $\theta_l$ is a path from the original root to this root node, by Theorem \ref{theorem4}, the reliability can be written as $B\left(v_{l,m}\right)=a_0+f_H\left(\theta_l\right)$. Generally, the pruned path $\zeta_n$ can be decomposed into two partial paths, that is, $\zeta_{n}=\left(\psi_{l-1},\chi_{n-l+1}\right)$, where the partial path $\psi_{l-1}$ is a path from the original root $v_{0,1}$ to the node $v_{l-1,\left\lceil \frac{m}{2} \right\rceil}$ and the partial path $\chi_{n-l+1}$ is a path from the node $v_{l-1,\left\lceil \frac{m}{2} \right\rceil}$ to the pruned leaf node. By Lemma \ref{lemma5}, the path $\chi_{n-l+1}$ is an all-zero path and the branch between the node $v_{l-1,\left\lceil \frac{m}{2} \right\rceil}$ and the node $v_{l,m}$ is taken the value $1$. Hence, we have $f_H\left(\theta_l\right)=f_H\left(\psi_{l-1}\right)=f_H\left(\zeta_n\right)-f_H\left(\chi_{n-l+1}\right)=f_H\left(\zeta_n\right)-1-(n-l)$.
\end{IEEEproof}

\begin{theorem}\label{theorem6}
Given a RCPP code with a length $M=N-Q$, for the C0 mode, the puncturing table generated by the QUP algorithm will maximize the spectrum distance SD1.
\end{theorem}
\begin{IEEEproof}
For an arbitrary puncturing scheme under the C0 mode, the original code tree can be decomposed into a group of subtrees. Assuming that each subtree has a depth $l_j$ and $2^{l_j}$ leaf nodes, the depth satisfies $0 \le l_j\le n-1$ by Theorem \ref{theorem4}. So the code length $M$ can be expressed as
\begin{equation}\label{equ_code_length}
\sum\limits_{l_j} {2^{l_j}}{\alpha_{l_j}}=M
\end{equation}
where $\alpha_{l_j}=0,1,2,\cdots$ stands for the number of subtrees with the depth $l_j$.
Define a set $\mathcal{E}=\left\{l_j\left|\alpha_{l_j}\ne 0\right.\right\}$, whose elements are arranged in the ascending order, that is, $l_1\le l_2 \le \ldots \le l_{|\mathcal{E}|}$.
So we can calculate the SD1 over all subtrees by Lemma \ref{lemma9}, yielding
\begin{equation}\label{SD1_C0_equ}
\begin{aligned}
\mathbb{E}\left[d_H(\omega_n)\right]&=\sum\limits_{j = 1}^{|\mathcal{E}|} \sum\limits_{k = 0}^{l_j} {\frac{1}{M}}\binom{l_j}{k}k\alpha_{l_j}\\
                                                    &=\sum\limits_{j = 1}^{|\mathcal{E}|} \sum\limits_{k = 0}^{l_j} {\frac{2^{l_j}}{M}}{\frac{1}{2^{l_j}}}\binom{l_j}{k}k\alpha_{l_j}\\
                                                    &=\sum\limits_{j = 1}^{|\mathcal{E}|} {\frac{2^{l_j}}{M}}\alpha_{l_j} \sum\limits_{k = 0}^{l_j}{\frac{1}{2^{l_j}}}\binom{l_j}{k}k
                                                    =\sum\limits_{j = 1}^{|\mathcal{E}|} {\frac{2^{{l_j}-1}}{M}}{l_j}\alpha_{l_j}. %\overset{(1)}{=}
\end{aligned}
\end{equation}
Just as the proof of Corollary \ref{corollary2}, the last equality is derived from the mean of binomial distribution.

For any puncturing scheme, the number $\alpha_{l_j}$ can be an arbitrary integer, such as $\alpha_{l_j}=0,1,2,\cdots$. We can treat the structure of SD1 in (\ref{SD1_C0_equ}) as a representation and carry of binary number from low-order to high-order. Assuming two consecutive orders $l_{j-1}$ and $l_j$ ($l_j\ge l_{j-1}+1$) and the corresponding digits $\alpha_{l_{j-1}}=2$ and $\alpha_{l_{j}}=1$, it is easy to assert
\begin{equation}
{(l_{j-1})}{2^{{l_{j-1}}-1}} {\alpha_{l_{j-1}}(=2)} < {l_j}{2^{l_j}}{\alpha_{l_j}(=1)}.
\end{equation}
Hence, in order to maximize SD1, the digits $\alpha_{l_j}$ should be limited to $0$ or $1$. This means that $(\alpha_{n-1},\cdots, \alpha_{0})$ is the binary expansion of the code length $M$.

On the other hand, according to lemma \ref{lemma8}, the digits $\alpha_{l_j}$ corresponding the QUP algorithm are taken the values $0$ or $1$. Therefore, the puncturing table of QUP algorithm can maximize the SD1.
\end{IEEEproof}

%\vspace{-1.686em}
\begin{theorem}\label{theorem7}
Given a RCPP code with a length $M=N-Q$, for the C0 mode, the puncturing table generated by the QUP algorithm will maximize the spectrum distance JSD.
\end{theorem}
\begin{IEEEproof}
Like the proof of Theorem \ref{theorem6}, for an arbitrary puncturing scheme under the C0 mode, the code length can be expanded by (\ref{equ_code_length}). For the subtrees with the same depth $l_j$, define $\mathcal{G}_j=\left\{\psi_{j,s}\left|s=1,2,\cdots,\alpha_{l_j}\right.\right\}$ as a pruned path set, where $\psi_{j,s}$ is the $s$-th pruned path containing a predecessor of one subtree. Let $n_{j,s}=f_H\left(\psi_{j,s}\right)-1$.

Thus, according to Lemma \ref{lemma9}, the spectrum distance SD0 can be calculated by averaging over all subtrees to yield
\begin{equation}\label{SD0_C0_equ}
\begin{aligned}
&\mathbb{E}\left[f_H(\omega_n)\right]
                  =\sum\limits_{j = 1}^{|\mathcal{E}|} \sum\limits_{s = 1}^{\alpha_{l_j}} \sum\limits_{r = 0}^{l_j} {\frac{1}{M}}\binom{l_j}{r}\left(r+n_{j,s}-l_j\right)\\
                  &=\sum\limits_{j = 1}^{|\mathcal{E}|} \sum\limits_{s = 1}^{\alpha_{l_j}} \sum\limits_{r = 0}^{l_j} {\frac{1}{M}}\binom{l_j}{r} n_{j,s}
                     -\sum\limits_{j = 1}^{|\mathcal{E}|} {\frac{2^{{l_j}-1}}{M}}{l_j}\alpha_{l_j}\\
                  &=\sum\limits_{j = 1}^{|\mathcal{E}|} {\frac{2^{l_j}}{M}}\sum\limits_{s = 1}^{\alpha_{l_j}}n_{j,s}-\mathbb{E}\left[d_H(\omega_n)\right] %\overset{(1)}{=}
\end{aligned}
\end{equation}
where the second line is derived from the mean of the binomial distribution. The above expression enables us to explicitly represent the JSD as
\begin{equation}\label{JSD_C0_equ}
\mathbb{E}\left[d_H(\omega_n)\right]+\mathbb{E}\left[f_H(\omega_n)\right]=\sum\limits_{j = 1}^{|\mathcal{E}|} {\frac{2^{l_j}}{M}}\sum\limits_{s = 1}^{\alpha_{l_j}}n_{j,s}.
\end{equation}

We also treat the structure of JSD in (\ref{JSD_C0_equ}) as a process of binary carry computation from low-order to high-order. Assuming two consecutive orders $l_{j}$ and $l_{j+1}$, if $\alpha_{l_j}=2$ for an arbitrary puncturing, there are two digits, $n_{j,1}$ and $n_{j,2}$, for two subtrees. On the contrary, if QUP scheme is applied, there is only one subtree and the corresponding digit is $n'_{j+1,1}=\max\left\{n_{j,1},n_{j,2}\right\}$. Obviously, we have
\begin{equation}\label{equ_JSD_C0_iter}
2^{l_j}\left(n_{j,1}+n_{j,2}\right) \le 2^{l_{j+1}}n'_{j+1,1}\left(=\max\left\{n_{j,1},n_{j,2}\right\}\right).
\end{equation}
Therefore, the puncturing table generated by QUP algorithm can maximize the JSD by a recursion of (\ref{equ_JSD_C0_iter}).
\end{IEEEproof}

For puncturing with the QUP algorithm, the code length $M$ is expressible as
\begin{equation}\label{equ_code_length_QUP}
\sum\limits_{j = 1}^{|\mathcal{F}|} {2^{l_j}}=M
\end{equation}
where $\mathcal{F}=\left\{l_j\left|\alpha_{l_j}=1\right.\right\}$ satisfies $l_1\le l_2 \le \cdots \le l_{|\mathcal{F}|}$.

\begin{corollary}
The PWEF1 of a RCPP code constructed by QUP algorithm is $\mathcal{H} (X)=\sum\limits_{j=1}^{|\mathcal{F}|} (1+X)^{l_j}
                    =\sum\limits_{j=1}^{|\mathcal{F}|}\sum\limits_{k=0}^{l_j} \binom{l_j}{k}X^{k}$.
\end{corollary}

\begin{theorem}\label{theorem8}
Suppose a RCPP code with the length $M=N-Q$, for the C0 mode, the SD1 corresponding to QUP satisfies $\frac{n-2}{2}\le \mathbb{E}\left[d_H(\omega_n)\right]\le \frac{n-1}{2}$.
\end{theorem}
\begin{IEEEproof}
First, we prove the right-side inequality. Due to $l_j \le n-1$, we have
\begin{equation}
\begin{aligned}
\sum\limits_{j = 1}^{|\mathcal{F}|} {\frac{2^{{l_j}-1}}{M}}{l_j} &\le \sum\limits_{j = 1}^{|\mathcal{F}|} {\frac{2^{{l_j}-1}}{M}}{(n-1)}\\
                                                                                                 &= \frac{n-1}{2} \sum\limits_{j = 1}^{|\mathcal{F}|} {\frac{2^{{l_j}}}{M}}=\frac{n-1}{2}.
\end{aligned}
\end{equation}
Next, let $S_0=\sum\limits_{j = 1}^{|\mathcal{F}|} {\frac{2^{{l_j}-1}}{M}}{l_j}=\sum\limits_{j = 1}^{|\mathcal{F}|-1} {\frac{2^{l_j}}{2M}}{l_j}+\frac{(n-1)2^{n-1}}{2M}$ and $S_1=\sum\limits_{j = 1}^{|\mathcal{F}|} {\frac{(n-2)2^{l_j}}{2M}}=\frac{n-2}{2}$, we need to prove $S_0>S_1$. Due to $l_{|\mathcal{F}|}=n-1$, we have
\begin{equation}
\begin{aligned}
S_0-S_1&=\frac{1}{2M}\left[2^{n-1}-\sum\limits_{j = 1}^{|\mathcal{F}|-1} \left(n-2-l_j\right)2^{l_j}\right]\\
           &\ge\frac{1}{2M}\left[2^{n-1}-\sum\limits_{k = 0}^{n-2} \left(n-2-k\right)2^k\right]\\
           &\overset{(1)}{=}\frac{1}{2M}\left[2^{n-1}-\left(2^{n-1}-n\right)\right]=\frac{n}{2M}>0
\end{aligned}
\end{equation}
where the equality (1) is derived from the summation of arithmetico-geometric sequence \cite{Book_Riley}.
\end{IEEEproof}

\begin{corollary}\label{corollary8}
For the C0 mode, the JSD corresponding to QUP satisfies ${n-2} \le d_{avg}+\lambda_{avg} \le {n-1}$.
\end{corollary}
\begin{IEEEproof}
By Theorem \ref{theorem7}, for the QUP puncturing, there is only one prune path corresponding to a subtree with a depth $l_j$ and this path contains the leftmost path of the parent subtree. So we have $l_j\le n_{j,1} \le n-1$. Like the proof in Theorem \ref{theorem8}, we have
$\sum\limits_{j = 1}^{|\mathcal{F}|} {\frac{2^{l_j}}{M}} n_{j,1} \le \sum\limits_{j = 1}^{|\mathcal{F}|} {\frac{2^{l_j}}{M}}{(n-1)}=n-1$.

For the left-side inequality, by Theorem \ref{theorem8}, we have
$\sum\limits_{j = 1}^{|\mathcal{F}|} {\frac{2^{l_j}}{M}} n_{j,1} \ge \sum\limits_{j = 1}^{|\mathcal{F}|} {\frac{2^{l_j}}{M}} {l_j}
                                                                                          =2\mathbb{E}\left[d_H(\omega_n)\right]\ge n-2$.\end{IEEEproof}

\subsection{Equivalent Class}
Recall that for single bit puncturing under the C0 mode, any code bit puncturing is equivalent by Theorem \ref{theorem3}. Generally, we have the following definition about the equivalent class for multiple bit puncturing.
\begin{definition}
Given the puncturing set of source bits $\mathcal{D}$ and its corresponding puncturing table $\mathscr{T}_N$ (or puncturing set of code bits $\mathcal{B}$), if another puncturing table $\mathscr{T}_N'$ (or $\mathcal{B'}$) can generate the same set $\mathcal{D}$, we call these two tables $\mathscr{T}_N$ and $\mathscr{T}_N'$ (or two sets $\mathcal{B}$ and $\mathcal{B'}$) are equivalent, that is, they belong to an equivalent class.
\end{definition}

For the QUP algorithm, the puncturing length $Q$ can be expressed as
\begin{equation}\label{equ_punc_length_QUP}
\sum\limits_{z = 1}^{|\mathcal{U}|} {2^{m_z}}=Q
\end{equation}
where $\mathcal{U}=\left\{m_z \right\}$ and satisfies $m_1\le m_2 \le \ldots \le m_{|\mathcal{U}|}$. Further, let $m_0=-\infty$ and $m_{|\mathcal{U}|+1}=n$. We introduce the function $h(x)=2^x$ to simplify the analysis.

\begin{theorem}
For the C0 mode, the number of puncturing tables equivalent to that generated by QUP algorithm is $h\left(\sum \nolimits_{z=1}^{|\mathcal{U}|} \left(n-2|\mathcal{U}|+2z-m_z\right)2^{m_z}\right)$.
\end{theorem}
\begin{IEEEproof}
Given a dual trellis with the original length $N=2^n$, we use $s_{p,q}, p=1,2,\cdots,n,q=0,1,\cdots,n$ to denote a variable node at the $p$-th row and the $q$-th column, where the row index is ascending from top to bottom and the column index is increasing from left to right.

By Lemma \ref{lemma8}, the source puncturing set generated by the QUP algorithm is $\mathcal{D}=\left\{1,2,\cdots,Q\right\}$. Let $E_{z}=\sum\nolimits_{o=z+1}^{|\mathcal{U}|} 2^{m_o}$ and $E_{|\mathcal{U}|}=0$. The set $\mathcal{D}$ can be decomposed into a group of subsets, that is, $\mathcal{D}=\bigcup \nolimits_{z=1}^{|\mathcal{U}|} \mathcal{D}_z$, where $\mathcal{D}_z=\left\{i|i=E_{z}+1, \cdots, E_{z}+2^{m_z}\right\}$. By this decomposition, we can separately consider the number of equivalent tables corresponding to each subset.

Suppose there are $J_z=\sum\nolimits_{e=0}^{z-1} 2^{m_e}$ source bits have been punctured, these nodes will generate $\xi_{z-1}$ equivalent puncturing schemes with $J_z$ candidate puncturing nodes at column $m_z$ due to the iterative application of Lemma \ref{lemma2}. At the present, we calculate the number of equivalent puncturings corresponding to the source bits in $\mathcal{D}_z$.

Define the relative node set on the trellis as $\left\{s_{p,0}\left|p\in\mathcal{D}_z\right.\right\}$ and the corresponding source vector as $u_{\mathcal{D}_z}$. Let $N_z=2^{m_z}$. After extending these nodes from column $0$ to column $m_z$, we can obtain a local code vector $c_1^{N_z}$ which satisfies $u_{\mathcal{D}_z}\mathbf{F}_2^{\otimes {m_z}}=c_1^{N_z}\mathbf{B}_{N_z}$. So the corresponding nodes in the set $\Lambda_{z}=\left\{s_{p,m_z}\left|p\in\mathcal{D}_z\right.\right\}$ become fully dependent by this local coding constraint. In order to puncture the nodes in $\mathcal{D}_z$, the nodes in $\Lambda_{z}$ are inevitably punctured.

In all $\xi_{z-1}$ equivalent puncturing schemes, we consider each constraint between the candidate nodes of one scheme and the nodes in $\Lambda_{z}$. Without loss of generality, the candidate nodes in the set $\Xi_{z}=\left\{s_{p,m_z}|p\in\bigcup\nolimits_{e=1}^{z-1} \mathcal{D}_e\right\}$ are chosen to form a multi-butterfly constraint with the nodes in $\Lambda_{z}$.

Due to $J_z < 2^{m_z}$, we have $\left|\Xi_{z}\right|<\left|\Lambda_{z}\right|$. Let $\Phi_{z}=\left\{s_{p,m_z}\left|p=E_z+1,\ldots,E_z+J_z\right.\right\}$, we have $\Phi_{z} \subset \Lambda_{z}$. When extending from column $m_z$ to column $m_z+1$, the nodes in $\Phi_{z}$ can form a multi-butterfly constraint with those in $\Xi_{z}$, that is,
\begin{equation}\label{equ_trellis_constraint}
\left\{
\begin{aligned}
&s_{p,\left(m_z+1\right)}                 =s_{p,m_z}\oplus s_{\left(p+2^{m_z}\right),m_z}\\
&s_{\left(p+2^{m_z}\right),\left(m_z+1\right)}  =s_{\left(p+2^{m_z}\right),m_z}\\
\end{aligned}
\right.
\end{equation}
where $s_{p,m_z}\in \Phi_{z}$ and $s_{\left(p+2^{m_z}\right),m_z} \in \Xi_{z}$. While, the nodes in $\Phi_{z}^c=\Lambda_{z}-\Phi_{z}$ are free and not constrained by the set $\Xi_{z}$.

Hence, we consider the equivalent puncturing nodes corresponding to two sets $\Phi_{z}$ and $\Phi_{z}^c$ respectively. In the first case, the node $s_{p,m_z}\in \Phi_{z}$ is a mandatory puncturing node and the node $s_{\left(p+2^{m_z}\right),m_z} \in \Xi_{z}$ only is a candidate one. In order to ensure these two nodes are punctured, by Lemma \ref{lemma2}, the generated nodes $s_{p,\left(m_z+1\right)}$ and $s_{\left(p+2^{m_z}\right),\left(m_z+1\right)}$ must be punctured. When extending from column $m_z+1$ to column $m_{z+1}$, each one of these two generated nodes can be regarded as a root of a tree with a depth $\left(m_{z+1}-m_{z}-1\right)$. Since only one node is punctured on each tree, by Theorem \ref{theorem3}, the number of equivalent puncturing nodes is $2^{\left(m_{z+1}-m_{z}-1\right)}$. Therefore, the total number of this case is ${\xi_{z}^{1}}=h\left({2 J_z \left(m_{z+1}-m_{z}-1\right)}\right)$.

In the second case, the node in $\Phi_{z}^c$ is a mandatory puncturing node, which can also be regarded as a root of a tree with a depth $\left(m_{z+1}-m_{z}\right)$. Similarly by Theorem \ref{theorem3}, the total number of equivalent schemes for this case is ${\xi_{z}^{2}}=h\left({{\left(2^{m_z}-J_z\right)} \left(m_{z+1}-m_{z}\right)}\right)$. So the number of equivalent puncturings for the source bits in $\bigcup\nolimits_{e=1}^{z} \mathcal{D}_e$ is $\xi_z=\xi_{z-1}\xi_z^1\xi_z^2$.

Iteratively applying the above analysis for all subsets, the number of equivalent puncturing tables of QUP algorithm is calculated by
\begin{equation}\label{equ_punc_class}
\begin{aligned}
\xi&=\prod \limits_{z=1}^{|\mathcal{U}|}\xi_z^1 \xi_z^2\\
    &=\prod \limits_{z=1}^{|\mathcal{U}|} h\left[\sum \limits_{e=0}^{z-1} 2^{m_e} \left(m_{z+1}-m_{z}-2\right)+2^{m_z} \left(m_{z+1}-m_{z}\right)\right]\\
    &=h\left[\sum \limits_{z=1}^{|\mathcal{U}|} \left(\sum \limits_{e=0}^{z} 2^{m_e} \left(m_{z+1}-m_{z}\right)-2\sum \limits_{e=0}^{z-1} 2^{m_e} \right)\right].
\end{aligned}
\end{equation}

The second term of the argument inside the function $h(\cdot)$ of (\ref{equ_punc_class}) can be rewritten as
\begin{equation}\label{hfun_2nd_sum}
\begin{aligned}
&2\sum \limits_{z=1}^{|\mathcal{U}|}\sum \limits_{e=0}^{z-1} 2^{m_e}\\
&=2|\mathcal{U}| 2^{m_0}+2(|\mathcal{U}|-1) 2^{m_1}+\cdots+2\cdot 2^{m_{|\mathcal{U}|-1}}\\
&=2\sum \limits_{z=1}^{|\mathcal{U}|}(|\mathcal{U}|-z) 2^{m_z}
\end{aligned}
\end{equation}
where $2^{m_0}=0$ due to $m_0=-\infty$.

Expanding the first term of the argument inside the function $h(\cdot)$ of (\ref{equ_punc_class}) and by $m_{|\mathcal{U}|+1}=n$, we have
\begin{equation}\label{hfun_1st_sum}
\begin{aligned}
&\sum \limits_{z=1}^{|\mathcal{U}|} \left[\sum \limits_{e=0}^{z} 2^{m_e} \left(m_{z+1}-m_{z}\right)\right]\\
&=\sum \limits_{e=0}^{1} 2^{m_e} \left(m_2-m_1\right)+\cdots+\sum \limits_{e=0}^{|\mathcal{U}|} 2^{m_e} \left(m_{|\mathcal{U}|+1}-m_{|\mathcal{U}|}\right)\\
&=-\sum \limits_{z=1}^{|\mathcal{U}|} m_z 2^{m_z}+m_{|\mathcal{U}|+1} \sum \limits_{e=0}^{|\mathcal{U}|} 2^{m_e}\\
&=\sum \limits_{z=1}^{|\mathcal{U}|} \left(n-m_z\right) 2^{m_z}.
\end{aligned}
\end{equation}

Combining (\ref{hfun_2nd_sum}) and (\ref{hfun_1st_sum}), we complete the proof.
\end{IEEEproof}

\begin{example}
An equivalent class example for the QUP puncturing under the C0 mode is shown in Fig. \ref{figure_equclass_puncturing_C0mode}. Given a dual trellis with the original code length $N=8$ and the punctured bits number $Q=3$, the puncturing set of source bits is $\mathcal{D}=\left\{1,2,3\right\}$. Due to $Q=3=2^1+2^0$, we have $m_0=-\infty, m_1=0, m_2=1,m_3=3$ and $\mathcal{D}_2=\left\{1,2\right\}, \mathcal{D}_1=\{3\}$.

As shown in Fig. \ref{figure_equclass_puncturing_C0mode}, since the nodes $\{s_{1,0}, s_{2,0}\}$ and $\{s_{1,1}, s_{2,1}\}$ compose a butterfly constraint (marked by a blue dash box), $s_{1,1}$ and $s_{2,1}$ are the mandatory puncturing nodes. On the other hand, the node $s_{3,0}$ has two candidate puncturing nodes $s_{3,1}$ and $s_{4,1}$ (marked by a green cross). Therefore, for the puncturing nodes in $\mathcal{D}_1$, the number of equivalent puncturings is $\xi_1=2^{2^{m_1}(m_2-m_1)}=2$. Assuming the candidate node $s_{3,1}$ is selected to be punctured, then the nodes $s_{1,2}$ and $s_{3,2}$ must be punctured because these two nodes form a butterfly constraint with the nodes $s_{1,1}$ and $s_{3,1}$.

Hence, the number of equivalent puncturings for the node $s_{2,1}$ is $\xi_2^2=2^{\left(2^{m_2}-2^{m_1}\right)\left(m_3-m_2\right)}=4$, which is corresponding to a perfect tree with the root $s_{2,1}$ (marked by red lines). Moreover, the number of equivalent puncturings for the nodes $s_{1,2}$ and $s_{3,2}$ is $\xi_2^1=2^{2\times2^{m_1} \left(m_3-m_2-1\right)}=4$, which is relative to two perfect trees with these two nodes as the roots (marked by purple and yellow lines respectively). So the total number of equivalent puncturing tables is $\xi=4\times4\times2=32$.

As an example, the puncturing set corresponding to QUP algorithm is $\mathcal{B}=\{1,3,5\}$. There are two sets with equivalent puncturing, that is, $\{\{1,2\},\{3,4\},\{5,6,7,8\}\}$ and $\{\{1,2,3,4\},\{5,6\},\{7,8\}\}$. We can arbitrarily select three indices from each set of these two set and form an equivalent puncturing scheme, such as $\{2,4,6\}$ or $\{2,6,8\}$.
\end{example}
\begin{figure}[h]
  \centering
  \includegraphics[width=0.7\columnwidth]{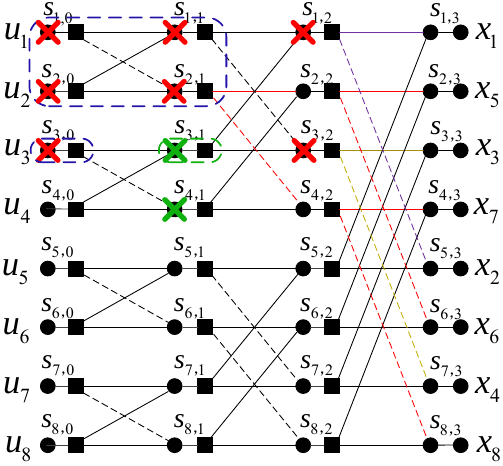}
  \caption{An example of QUP equivalent class on a dual trellis with $N=8$ and $Q=3$}
  \label{figure_equclass_puncturing_C0mode}
\end{figure}

\section{Optimal Puncturing For the Capacity-One Mode}
\label{section_VI}
In this section, reversal quasi-uniform puncturing (RQUP)  is described and proved to maximize the SD0 and JSD.
\subsection{Single Bit Puncturing}
Single bit puncturing under the C1 mode is symmetric to the operation under the C0 mode. Therefore, the corresponding polar spectra are also almost optimal.

\begin{theorem}\label{theorem11}
For single bit puncturing under the C1 mode, the maximal complemental path weight on the code tree is $n-1$, in addition, the maximal path weight is $n-1$.
\end{theorem}

Based on the symmetry of single bit puncturing under the C0 and C1 modes, the proof is similar to that of Theorem \ref{theorem4} and omitted.
\subsection{Reversal QUP algorithm}
For the C1 mode, the reversal quasi-uniform puncturing (RQUP) algorithm can be described as follows:
\medskip
\begin{enumerate}[ \textbf{Stage} 1)]
\item
Initialize the table $\mathscr{T}_N$ as all ones, and then set the last $Q$ bits of the vector as zeros;
\item
Perform bit-reversal permutation on the table $\mathscr{T}_N$ and obtain the puncturing table.
\end{enumerate}
\medskip

One of the main differences between RQUP and QUP algorithm is that the initialization table of RQUP is reversal to that of QUP. The puncturing table generated by the RQUP algorithm is also constructive and regular. It is easy to prove that RQUP has a similar property of Theorem \ref{theorem5}.

\begin{theorem}
RQUP algorithm can ensure that each punctured code bit is known by the decoder.
\end{theorem}
\begin{IEEEproof}
Let $\hat {x}_1^N=x_1^N\mathbf{B}_N$ denote the codeword after bit-reversal permutation. For the RQUP algorithm, the subvector $\hat {x}_{N-Q+1}^N$ should be punctured. Due to the coding constraint $u_1^N\mathbf{F}_2^{ \otimes n}=\hat{x}_1^N$ and the lower-triangle property of the matrix $\mathbf{F}_2^{ \otimes n}$, if each source bit $u_{N-j},j=0,1,\cdots,Q-1$ is set to a frozen value, each bit $\hat{x}_{N-j}$ will be set to a fixed value.
\end{IEEEproof}

\begin{corollary}\label{corollary10}
For the RQUP algorithm, the source bit vector $u_{N-Q+1}^N$ is punctured, that is, $\mathcal{D}=\left\{N-Q+1,\cdots,N\right\}$. Equally, the $Q$ rightmost leaf nodes on the code tree are pruned.
\end{corollary}

\subsection{Optimal Puncturing Table}
Like the RCPP code design under the C0 mode, we also concern the puncturing scheme under the C1 mode to optimize the spectrum distance.
\begin{lemma}\label{lemma11}
For a subtree after any puncturing scheme under the C1 mode, suppose $v_{l,m}$ and $v_{l-1,\left\lceil \frac{m}{2} \right\rceil}$ are the root node and its predecessor. Let $\zeta_n$ denote a pruned path containing the rightmost path of the subtree $\mathcal{T}\left(v_{l-1,\left\lceil \frac{m}{2} \right\rceil}\right)$. So the reliability of the root node $v_{l,m}$ can be expressed by $B\left(v_{l,m}\right)=2^{\left(d_H\left(\zeta_n\right)-1-(n-l)\right)}a_0$.
\end{lemma}

\begin{IEEEproof}
Like the proof of Lemma \ref{lemma9}, for a root node $v_{l,m}$ on a final subtree, assuming the corresponding path $\theta_l$ is a path from the original root to this root node, by Lemma \ref{C1_mode_UB_set}, the reliability can be written as $B\left(v_{l,m}\right)=2^{d_H\left(\theta_l\right)}a_0$. Further, like the definition in lemma \ref{lemma9}, the pruned path $\zeta_n$ can be decomposed into two partial paths, that is, $\zeta_{n}=\left(\psi_{l-1},\chi_{n-l+1}\right)$. By Lemma \ref{lemma6}, the path $\chi_{n-l+1}$ is an all-one path and the branch between the node $v_{l-1,\left\lceil \frac{m}{2} \right\rceil}$ and the node $v_{l,m}$ is taken the value $0$. Hence, we have $d_H\left(\theta_l\right)=d_H\left(\psi_{l-1}\right)=d_H\left(\zeta_n\right)-d_H\left(\chi_{n-l+1}\right)=d_H\left(\zeta_n\right)-1-(n-l)$.
\end{IEEEproof}

\begin{theorem}\label{theorem14}
Given a RCPP code generated under the C1 mode, the puncturing table generated by the RQUP algorithm will maximize the spectrum distance SD0.
\end{theorem}
\begin{IEEEproof}
For an arbitrary puncturing scheme under the C0 mode, the original code tree can be decomposed into a group of subtrees.

Using the binary expansion in (\ref{equ_code_length}) and the definition of set $\mathcal{E}$ in Theorem \ref{theorem6}, we can calculate and average the spectrum distance SD0 over all subtrees by Lemma \ref{lemma11}, that is, SD0 can be written as
\begin{equation}\label{SD0_C1_equ}
\begin{aligned}
\mathbb{E}\left[f_H(\omega_n)\right]=\sum\limits_{j = 1}^{|\mathcal{E}|} \sum\limits_{r = 0}^{l_j} {\frac{1}{M}}\binom{l_j}{r}r\alpha_{l_j}
                                                    =\sum\limits_{j = 1}^{|\mathcal{E}|} {\frac{2^{{l_j}-1}}{M}}{l_j}\alpha_{l_j}. %\overset{(1)}{=}
\end{aligned}
\end{equation}

Like the proof in Theorem \ref{theorem6}, in order to maximize SD0, the digit $\alpha_{l_j}$ should be limited to $0$ or $1$. This means that $(\alpha_{n-1},\ldots, \alpha_{0})$ is the binary representation of the code length $M$. Therefore, the RQUP algorithm can maximize SD0.
\end{IEEEproof}

\begin{theorem}
Given a RCPP code with a length $M=N-Q$, for the C1 mode, the puncturing table generated by the RQUP algorithm will maximize the spectrum distance JSD.
\end{theorem}
\begin{IEEEproof}
Like the proof of Theorem \ref{theorem7}, for an arbitrary puncturing scheme under the C1 mode, we introduce a pruned path set $\mathcal{G}_j=\left\{\psi_{j,s}\left|s=1,2,\cdots,\alpha_{l_j}\right.\right\}$ and a digit $n_{j,s}=d_H\left(\psi_{j,s}\right)-1$. So SD1 can be calculated and averaged over all subtrees by Lemma \ref{lemma11}, that is,
\begin{equation}\label{SD1_C1_equ}
\mathbb{E}\left[d_H(\omega_n)\right]
                  =\sum\limits_{j = 1}^{|\mathcal{E}|} {\frac{2^{l_j}}{M}}\sum\limits_{s = 1}^{\alpha_{l_j}}n_{j,s}-\mathbb{E}\left[f_H(\omega_n)\right]. %\overset{(1)}{=}
\end{equation}
By the same proof in Theorem \ref{theorem7}, we conclude that RQUP algorithm can maximize the JSD.
\end{IEEEproof}

\begin{corollary}
The PWEF0 of a RCPP code constructed by RQUP algorithm is
\begin{equation}
\mathcal{C} (X)=\sum\limits_{j=1}^{|\mathcal{F}|} (1+X)^{l_j}
                      =\sum\limits_{j=1}^{|\mathcal{F}|}\sum\limits_{r=0}^{l_j} \binom{l_j}{r}X^{r},
\end{equation}
\end{corollary}
where the set $\mathcal{F}$ is defined in (\ref{equ_code_length_QUP}).

\begin{theorem}\label{theorem16}
Suppose a RCPP code with a length $M=N-Q$, for the C1 mode, the SD0 and JSD corresponding to RQUP satisfies $\frac{n-2}{2}\le \mathbb{E}\left[f_H(\omega_n)\right]\le \frac{n-1}{2}$ and ${n-2} \le d_{avg}+\lambda_{avg} \le {n-1}$ respectively.
\end{theorem}
The proof is similar to that of Theorem \ref{theorem8} and Corollary \ref{corollary8}.

\section{Numerical Analysis and Simulation Results}
\label{section_VII}
In this section, at first, we compare various puncturing schemes under the C0 or C1 modes by calculating the spectrum distances SD0/SD1. Then RCPP codes based on different puncturing schemes under SC or SCL decodings are evaluated. Furthermore, the BLERs of RCPP and turbo codes are also compared via simulations over AWGN channels.
\subsection{Numerical Analysis of Puncturing Schemes}
We compare the spectrum distances of various puncturing schemes. For the C0 mode, we mainly concern three typical puncturing schemes, such as QUP algorithm \cite{RCPP_Niu}, the algorithm proposed by Eslami \emph{et al.} \cite{Finite_Eslami} and that proposed by Shin \emph{et al.} \cite{LCPC_Shin}. On the other hand, for the C1 mode, we mainly investigate two puncturing schemes, such as RQUP algrithm proposed in this paper and the algorithm proposed by Wang \emph{et al.} \cite{Novel_Punc}. For the latter, given the generator $\mathbf{G}_N$, the index of column with column weight 1 is selected as the punctured position. However there may be many selections for the column weight 1 as stated in \cite[Algorithm1]{Novel_Punc}. In order to simplify evaluation, we use a puncturing table where the last $Q$ code bits are punctured as a reference of Wang algorithm.

\begin{figure}[h]
  \centering
  \includegraphics[width=1\columnwidth]{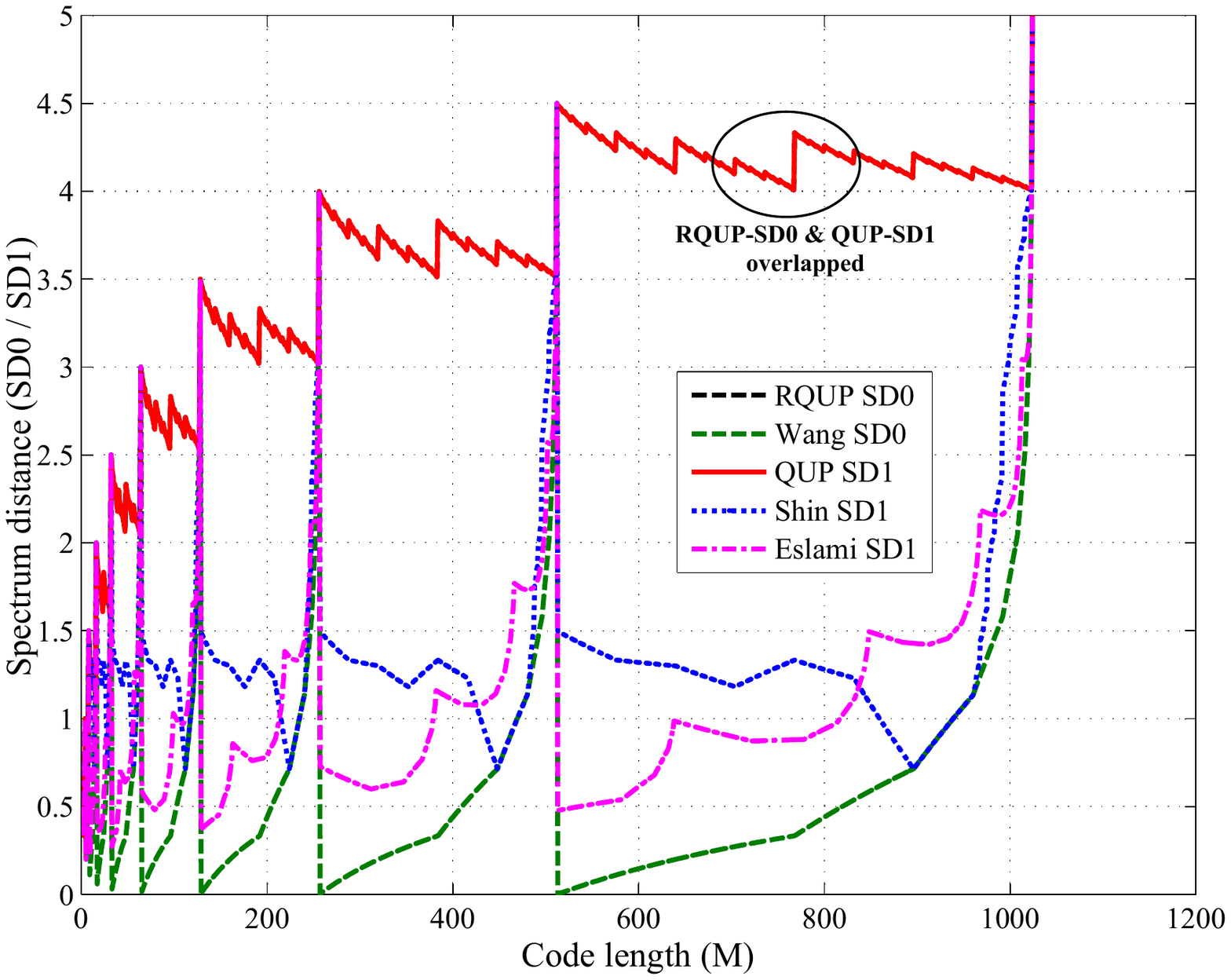}
  \caption{Spectrum distances (SD1 or SD0) for different puncturing schemes under the C0 or C1 modes}
  \label{fig_SD0_SD1_vs_code_length}
\end{figure}

For all the puncturing schemes under the C0 or C1 modes, the spectrum distances SD1/SD0 versus code length ($M=1\sim1024$) are shown in Fig. \ref{fig_SD0_SD1_vs_code_length}. Among the three puncturing schemes (QUP/Shin/Eslami) under the C0 mode, the SD1 of QUP algorithm is larger than that of the others due to the optimal polar spectra PS1. Similarly, the SD0 of RQUP is better than that of Wang method due to the optimal PS0. Recall that the polar spectra of QUP and RQUP schemes are symmetrical, the SD1 of QUP and SD0 of RQUP are overlapped as depicted in Fig. \ref{fig_SD0_SD1_vs_code_length}. Further, we observe that the SD1 (SD0) of QUP (RQUP) is distributed between $\frac{n-2}{2}$ and $\frac{n-1}{2}$ which is consistent with Theorem \ref{theorem8} (\ref{theorem16}).

For JSDs of all the puncturing schemes, we can observe the similar results, that is, QUP or RQUP have the maximal JSDs under the C0 or C1 modes. Due to the limitation of space, these results are not shown here. However, the performance comparison just based on JSD may result a bias conclusion. As an example, SD1 versus SD0 at the code length $M=990\sim1024$ for all the schemes is drawn in Fig. \ref{fig_SD1_vs_SD0}.

\begin{figure}[h]
  \centering
  \includegraphics[width=1\columnwidth]{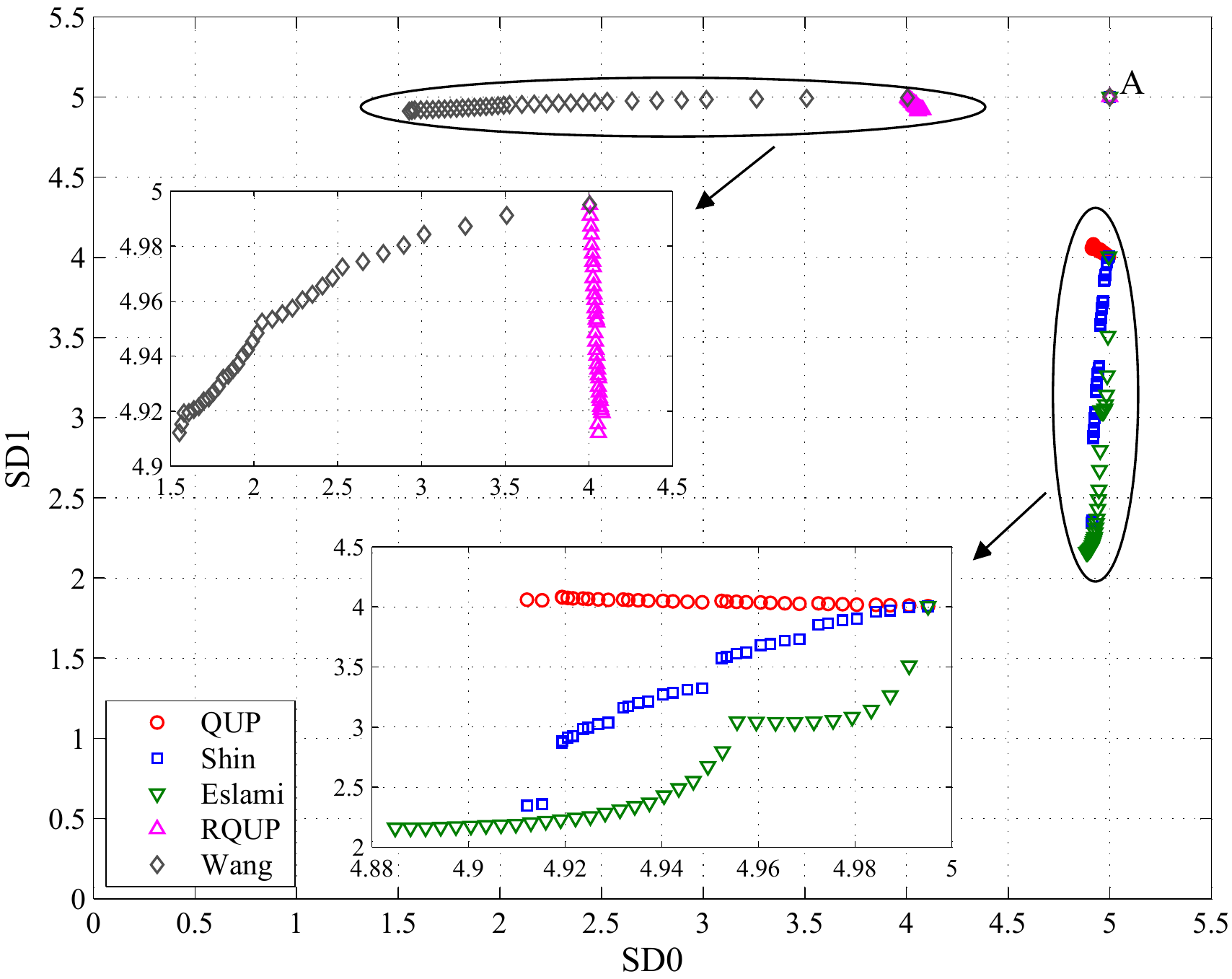}
  \caption{Spectrum distance (SD1) vs Spectrum distance (SD0)}
  \label{fig_SD1_vs_SD0}
\end{figure}

In this $2$-D chart, the point A located at $(5,5)$ is relative to the SD1/SD0 of the original polar code with the code length $1024$. Recall that the aim of RCPP codes optimization is to approach the spectrum distances of the parent codes as close together as possible, that is, in this chart, the more one point corresponding to a puncturing scheme is close to the point A, the better this scheme will achieve an error performance. All points relative to QUP are concentrated at $(5,4)$ and all points relative to RQUP at $(4,5)$. Obviously, among the three puncturing schemes (QUP/Shin/Eslami) under the C0 mode, QUP has the maximal value of SD1 when the value of SD0 is fixed. On the other hand, given the fixed SD1, the SD0 of QUP is larger than that of Wang scheme.

\subsection{Simulation Results}
First, we compare the error performance of RCPP codes with various puncturing schemes under the BI-AWGN channels. The Gaussian approximation algorithm \cite{GA_Trifonov} is applied to construct these codes. Given the SC decoding and the parent code length $N=1024$, the BLER performance comparisons of RCPP codes based on all the puncturing schemes with the code length $M=700$ are shown in Fig. \ref{fig_varrate_SC_comp} for the code rate $R=1/3$, $R=3/4$ and $R=1/2$ respectively.

\begin{figure}[h]
  \centering
  \includegraphics[width=1.02\columnwidth]{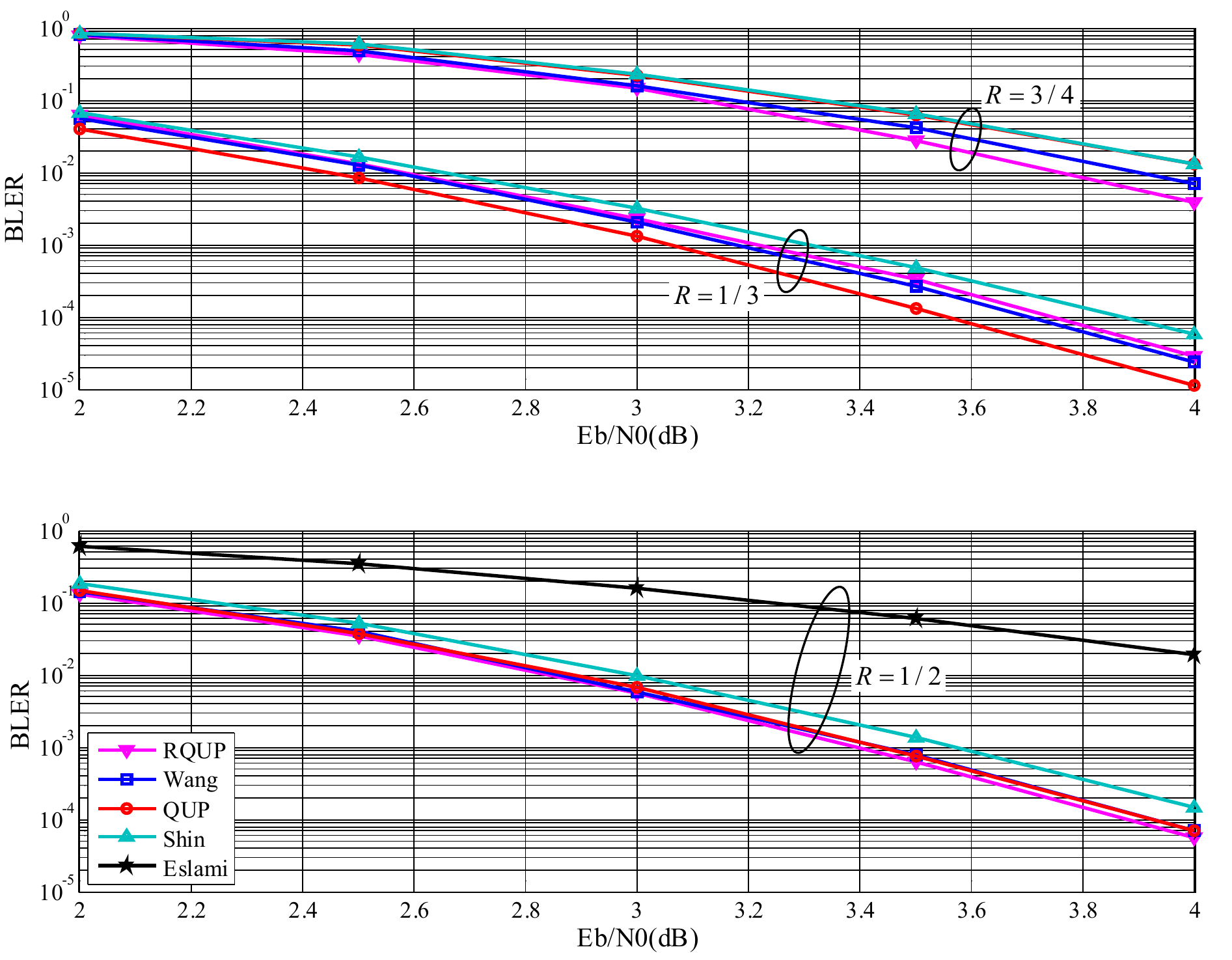}
  \caption{BLER performance comparisons under SC decoding for RCPP codes constructed by various puncturing schemes with the code length $M=700$ and code rate $R=1/3,1/2,3/4$ (the parent code length ${N}=1024$).}
  \label{fig_varrate_SC_comp}
\end{figure}
For the low code rate $R=1/3$, compared with other schemes, such as Wang, RQUP and Shin algorithms, we can see that QUP achieves the best error performance. On the other hand, for the high code rate $R=3/4$, RQUP is the best one among all the puncturing schemes. These results are consistent with the analysis in Section \ref{section_V} and \ref{section_VI}. Further, we find that the error performance of QUP is worse than that of RQUP in the high code rate and vice versa in the low code rate. Especially, we observe that the schemes of QUP, RQUP and Wang can achieve almost the same performance and they are better than Shin or Eslami schemes. These phenomena may imply that the code rate $R=1/2$ is a critical value. So RQUP will be the best scheme when $R>1/2$ and QUP will be the best one when $R<1/2$.

Next we compare the performance of RCPP and turbo codes under AWGN channel. RCPP codes are constructed from the parent code with the code length $N=1024$ by QUP or RQUP schemes and CA-SCL is used as a decoding algorithm with the maximum list size $32$. An eight-state turbo code in 3GPP LTE standard \cite{LTE_36212} is used as a reference. A CRC code is used in all concatenation coding schemes (both for turbo and RCPP codes). The Log-MAP algorithm is applied in turbo decoding and the maximum number of iterations is ${{I}_{\max }}=8$.

\begin{figure}[h]
  \centering
  \includegraphics[width=1.02\columnwidth]{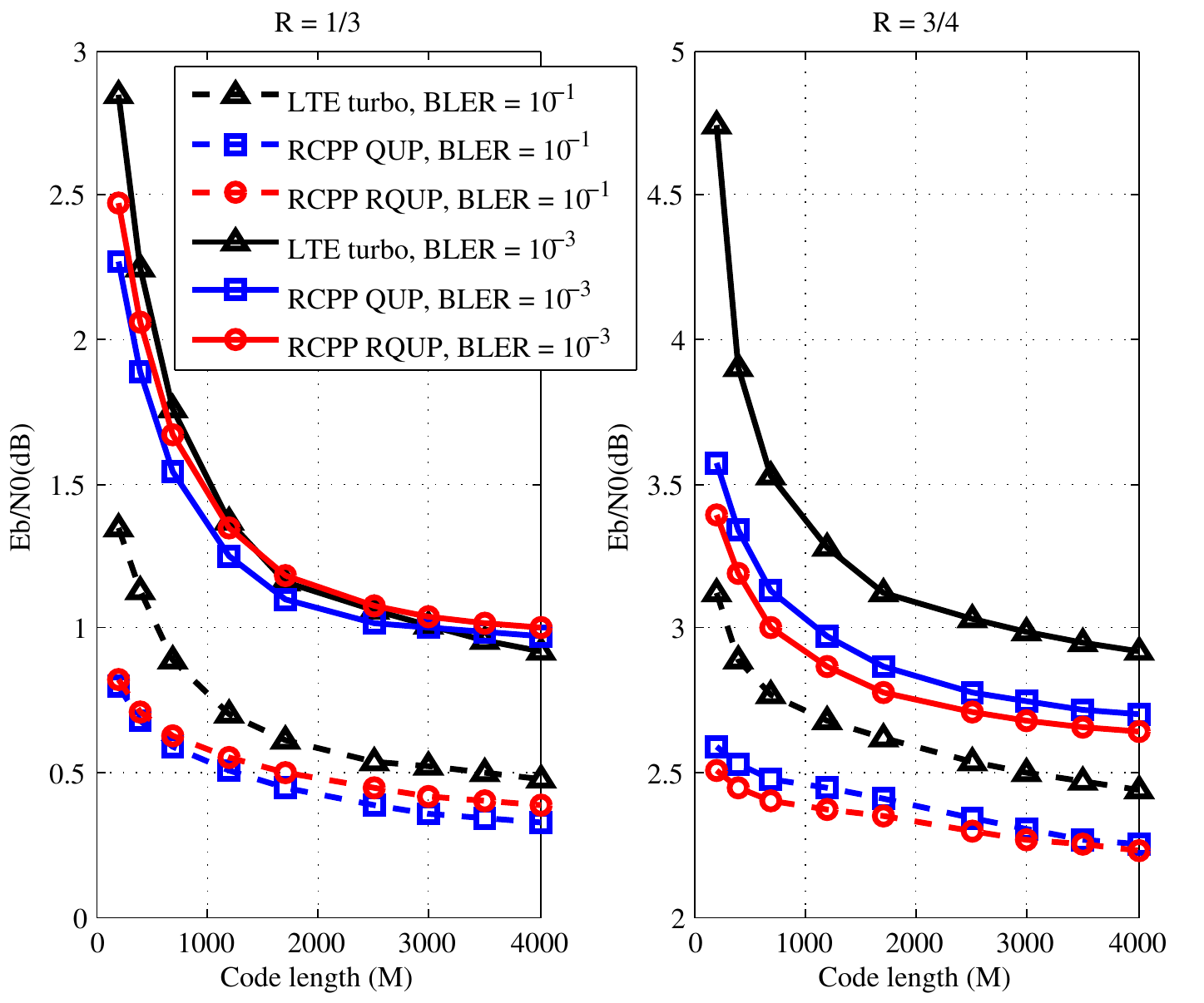}
  \caption{Eb/N0 vs code length of RCPP codes (punctured by QUP and RQUP algorithms) and LTE turbo codes at the code rate $R=1/3,3/4$ and BLER of ${{10}^{-1}},{{10}^{-3}}$}
  \label{fig_coding_gain33}
\end{figure}

We investigate the relationship of bit signal noise ratio (SNR) and code length for these two codes. The performance curves of $E_b/N_0$ vs code length $M$ ($200\thicksim4000$) for the LTE turbo and RCPP codes (punctured by QUP and RQUP algorithms) at the BLER of ${{10}^{-1}},{{10}^{-3}}$ and code rate $R=1/3,3/4$ are shown in Fig. \ref{fig_coding_gain33}.

In most cases, RCPP codes can achieve additional coding gains relative to LTE turbo codes. For the low code rate $R=1/3$, as shown in Fig. \ref{fig_coding_gain33}, a maximum $0.5$ dB additional gain can be obtained at the code length $M=200$ and the RCPP codes punctured by QUP algorithms can achieve slightly better performance than those codes punctured by RQUP. On the other hand, for the high code rate $R=3/4$, a maximum $1.35$ dB performance gain can be attained at the code length $M=200$ and the BLER of $10^{-3}$. In contrast to the case of low code rate, the RQUP algorithms can generate better RCPP codes in this case.

%\begin{figure}[h]
%  \centering
%  \includegraphics[width=0.8\columnwidth]{Fig19_SNR_vs_codelength_R75.pdf}
%  \caption{Eb/N0 vs code length of RCPP codes (QUP and RQUP) and LTE turbo codes at the code rate $R=3/4$ and BLER of ${{10}^{-1}}$, ${{10}^{-3}}$.}
%  \label{fig_coding_gain75}
%\end{figure}

For the medium code rate $R=1/2$, additional gain can be obtained and the RCPP codes punctured by QUP or RQUP algorithms can achieve the same performance. Due to the limitation of space, these results are not shown.

\section{Conclusions}
\label{section_VIII}

In this paper, we propose a theoretic framework based on the polar spectra to analyze and design rate-compatible punctured polar code. Guided by the spectrum distances, two simple quasi-uniform puncturing methods (QUP and RQUP) are proposed to generate the puncturing tables under the C0/C1 modes. By the analysis of the performance metrics, such as SD0/SD1/JSD, we prove that these two algorithms can achieve the maximal value of corresponding spectrum distance. Simulation results in AWGN channel show that the performance of RCPP codes by QUP or RQUP can be equal to or exceed that of the turbo codes at the same code length.

\section{Puncture Algorithm}

\begin{algorithm}[htbp]
\setlength{\abovecaptionskip}{0.cm}
\setlength{\belowcaptionskip}{-0.cm}
\caption{Optimal SF search algorithm}\label{SF_search_algorithm}
\KwIn {SCMA system parameters and the precision $\varepsilon $;}
\KwOut {Optimal SF $\alpha^*$;}
Initialize $a = 0$, $b = 1$, and calculate $\Theta '\left( a \right)$, $\Theta '\left( b \right)$\;
\While{$\left| {a - b} \right| > \varepsilon $}
{
    Calculate $c = \left( {a + b} \right)/2$\;
    Run the S-EXIT method to calculate $\Theta '\left( c \right)$\;
    \If{$\Theta '\left( a \right) \cdot \Theta '\left( c \right) \le 0$}
    {
        Let $b = c$; \tcc*[f]{Optimal SF $\alpha^*  \in \left( {a,c} \right]$}
    }
    \Else
    {
        Let $a = c$; \tcc*[f]{Optimal SF $\alpha^*  \in \left( {c,b} \right)$}
    }
}
Record $\alpha^*  = \left( {a + b} \right)/2$\;
\end{algorithm}

% conference papers do not normally have an appendix
\appendix
\subsection{Proof of Lemma \ref{lemma2}}
Without loss of generality, we consider the polarization of the first scenario $\mathscr{T}_2=(0,1)$, that is, $\left(\mathbb{W},W\right)\mapsto\left(\widetilde{W}_2^{(1)},\widetilde{W}_2^{(2)}\right)$. Under the C0 mode, due to $\forall y_1, \mathbb{W}\left(y_1\left|0\right.\right)=\mathbb{W}\left(y_1\left|1\right.\right)=\frac{1}{2}$, for $\forall y_1^2$, the transition probabilities of polarized channel $\widetilde{W}_2^{(1)}$ can be written by
\begin{equation}
\begin{aligned}
&\widetilde{W}_2^{(1)}\left(y_1^2\left|0\right.\right)=\sum \limits_{u_2}\frac{1}{2}\mathbb{W}\left(y_1\left| u_2 \right.\right) W\left(y_2\left|u_2\right.\right)\\
                                                                       &=\mathbb{W}\left(y_1\left| 0 \right.\right) \sum \limits_{u_2}\frac{1}{2} W\left(y_2\left|u_2\right.\right)
                                                                       =\widetilde{W}_2^{(1)}\left(y_1^2\left|1\right.\right).
\end{aligned}
\end{equation}
So the channel $\widetilde{W}_2^{(1)}$ is degraded to a punctured channel and $Z\left(\widetilde{W}_2^{(1)}\right)=1$.

On the other hand, we analyze the LLR of polarized channel $\widetilde{W}_2^{(2)}$.
Let $\mathbb{L}(y_2)=\ln \frac{W(y_2\left|0\right.)}{W(y_2\left|1\right.)}$ denote the LLR of B-DMC $W(y_2\left|u_2\right.)$. Due to $\mathbb{L}(y_1)=0$, the corresponding probability density function (PDF) is $F(\mathbb{L}(y_1))=\delta(y_1)$, where $\delta(\cdot)$ is the Dirac function. Since the source bit $u_2$ is relative to a variable node, the PDF of the corresponding LLR can be derived as
\begin{equation}
\begin{aligned}
F\left(\mathbb{L}\left(y_1^2,u_1\right)\right)&=F\left(\mathbb{L}\left(y_1\right)\right)\ast F\left(\mathbb{L}\left(y_2\right)\right)\\
                                                 &=\delta(y_1)\ast F\left(\mathbb{L}\left(y_2\right)\right)=F\left(\mathbb{L}\left(y_2\right)\right),
\end{aligned}
\end{equation}
where $\ast$ is the convolutional operation. So the polarized channel $\widetilde{W}_2^{(2)}$ has the same reliability as that of the original B-DMC, that is, $Z(\widetilde{W}_2^{(2)})=Z(W)=Z_0$.

For the second puncturing scheme, as shown in Fig. \ref{fig_two_channel}(b), based on the symmetric property of swapping two channels ($\mathbb{W}$ and $W$), we can conclude the same results.

\subsection{Proof of Lemma \ref{C1_mode_BC}}\label{proof_C1_mode_BC}
The Bhattacharyya parameters of the polarized channels $\widetilde {W}_N^{(i)}$ under the C1 mode can be written by
\begin{equation}
\begin{aligned}
&Z\left(\widetilde{W}_N^{(i)}\right)\\
&=\sum \limits_{y_1^N \in \mathcal{Y}^N} \sum \limits_{u_1^{i-1} \in \mathcal{X}^{i-1}}
\sqrt {\widetilde{W}_N^{(i)}\left(y_1^N,u_1^{i-1}\left|0\right.\right)\widetilde{W}_N^{(i)}\left(y_1^N,u_1^{i-1}\left|1\right.\right)}.
\end{aligned}
\end{equation}

By the coding relationship $u_1^N\mathbf{G}_N=x_1^N$, the channel transition probabilities can be presented as
\begin{equation}
\begin{aligned}
&\widetilde{W}_N^{(i)}\left(y_1^N,u_1^{i-1}\left|u_i\right.\right)\\
&=\frac{1}{2^{N-1}}\sum \limits_{x_{\mathcal{B}^c}} \prod \limits_{j\in \mathcal{B}^{c}} {W\left(y_j\left|x_j\right.\right)}\sum \limits_{x_\mathcal{B}} \prod \limits_{m\in \mathcal{B}} {\mathbb{W}\left(y_m\left|x_m\right.\right)},
\end{aligned}
\end{equation}
where the punctured vector $x_{\mathcal{B}}$ is composed of the punctured code bits $x_m$ ($m\in\mathcal{B}$), and the corresponding received vector can be written by $y_{\mathcal{B}}=\left\{y_m\right\}_{m\in \mathcal{B}}$.

Under the C1 mode, we assume $\mathbb{W}\left(\dot{y}_m|\dot{x}_m\right)=1$ is only true for each specific pair $\left(\dot{x}_m,\dot{y}_m\right)$. Therefore, for the specific vector $\dot{y}_{\mathcal{B}}$, we can write $\sum \nolimits_{x_\mathcal{B}} \prod \nolimits_{m\in \mathcal{B}} {\mathbb{W}\left(y_m\left|x_m\right.\right)}=\prod \nolimits_{m\in \mathcal{B}} {\mathbb{W}\left(\dot{y}_m\left|\dot{x}_m\right.\right)}=1$. Let $x_1^N=\left(u_1^{i-1},0,u_{i+1}^N\right)\mathbf{G}_N$ and ${x'}_1^N=\left(u_1^{i-1},1,{u'}_{i+1}^N\right)\mathbf{G}_N$. Furthermore, these two vectors satisfy $x_1^N=\left(x_{\mathcal{B}^c},\dot{x}_{\mathcal{B}}\right)$ and ${x'}_1^N=\left({x'}_{\mathcal{B}^c},\dot{x}_{\mathcal{B}}\right)$ respectively. So we have
\begin{equation}\label{C1_BC_Nchannel}
\begin{aligned}
&Z\left(\widetilde{W}_N^{(i)}\right)
=\sum \limits_{y_1^N \in \mathcal{Y}^M} \sum \limits_{u_1^{i-1} \in \mathcal{X}^{i-1}}\frac{1}{2^{N-1}}\\
&\cdot\sqrt {\sum \limits_{x_{\mathcal{B}^c}} \prod \limits_{j\in \mathcal{B}^{c}} {W\left(y_j\left|x_j\right.\right)}
           \sum \limits_{x'_{\mathcal{B}^c}} \prod \limits_{j\in \mathcal{B}^{c}} {W\left(y_j\left|x'_j\right.\right)}}
\end{aligned},
\end{equation}
where $y_1^N=\left(y_{\mathcal{B}^c},\dot{y}_{\mathcal{B}}\right)$.

On the contrary, if the vectors $x_{\mathcal{B}}$, ${x'}_{\mathcal{B}^c}$ and $y_{\mathcal{B}}$ can be arbitrarily selected, we have
\begin{equation}\label{prob_inequ}
\sum \limits_{y_{\mathcal{B}}} \sum \limits_{x_{\mathcal{B}}} \prod \limits_{m\in \mathcal{B}} {W\left(y_m\left|x_m\right.\right)} \sum \limits_{x'_{\mathcal{B}}} \prod \limits_{m\in \mathcal{B}} {W\left(y_m\left|x'_m\right.\right)}> 1.
\end{equation}

Substituting (\ref{prob_inequ}) into (\ref{C1_BC_Nchannel}), the Bhattacharyya parameters can be enlarged by
\begin{equation}
\begin{aligned}
&Z\left(\widetilde{W}_N^{(i)}\right)
< \sum \limits_{y_{\mathcal{B}^c} \in \mathcal{Y}^M} \sum \limits_{u_1^{i-1} \in \mathcal{X}^{i-1}}\frac{1}{2^{N-1}}\\
&\cdot\sqrt {\sum \limits_{x_{\mathcal{B}^c}} \prod \limits_{j\in \mathcal{B}^{c}} {W\left(y_j\left|x_j\right.\right)}
           \sum \limits_{x'_{\mathcal{B}^c}} \prod \limits_{j\in \mathcal{B}^{c}} {W\left(y_j\left|x'_j\right.\right)}}\\
&\cdot \sqrt {\sum \limits_{y_{\mathcal{B}}} \sum \limits_{x_{\mathcal{B}}} \prod \limits_{m\in \mathcal{B}} {W\left(y_m\left|x_m\right.\right)} \sum \limits_{x'_{\mathcal{B}}} \prod \limits_{m\in \mathcal{B}} {W\left(y_m\left|x'_m\right.\right)}}\\
&\leq\sum \limits_{y_1^N \in \mathcal{Y}^N} \sum \limits_{u_1^{i-1} \in \mathcal{X}^{i-1}}\frac{1}{2^{N-1}}\\
&\cdot\sqrt {\sum \limits_{x_{\mathcal{B}^c}}   \prod \limits_{j\in \mathcal{B}^{c}} {W\left(y_j\left|x_j\right.\right)}
                  \sum \limits_{x_{\mathcal{B}}}      \prod \limits_{m\in \mathcal{B}} {W\left(y_m\left|x_m\right.\right)} }\\
&\cdot\sqrt {\sum \limits_{x'_{\mathcal{B}^c}}  \prod \limits_{j\in \mathcal{B}^{c}} {W\left(y_j\left|x'_j\right.\right)}
                  \sum \limits_{x'_{\mathcal{B}}}     \prod \limits_{m\in \mathcal{B}} {W\left(y_m\left|x'_m\right.\right)} }\\
&=Z\left({W}_N^{(i)}\right).
\end{aligned}
\end{equation}

\subsection{Proof of Lemma \ref{lemma3}}
Recall that the value of the punctured bit in the C1 mode is known by the decoder. Apparently, puncturing the code bit $x_2$ is a good selection because this bit is only involved one source bit $u_2$. Hence, in order to ensure that the bit $x_2$ is punctured and has a fixed value to the decoder, as shown in Fig. \ref{fig_two_channel}(c), the puncturing table should be $\mathscr{T}_2=(1,0)$.

Let $\mathbb{L}(y_1^2)=\ln\frac{\widetilde{W}_2^{(1)}\left(y_1^2\left|0\right.\right)}{\widetilde{W}_2^{(1)}\left(y_1^2\left|1\right.\right)}$ and $\mathbb{L}(y_2)=\ln \frac{\mathbb{W}\left(y_2\left|0\right.\right)}{\mathbb{W}\left(y_2\left|1\right.\right)}=+\infty$ denote the LLRs of the source bits $u_1$ and  $u_2$ ($=0$) respectively. Considering the check node constraint, we have
\begin{equation}
\tanh\left(\frac{\mathbb{L}(y_1^2)}{2}\right)=\tanh\left(\frac{\mathbb{L}(y_1)}{2}\right)\cdot 1,
\end{equation}
where $\tanh(\cdot)$ is the hyperbolic tangent function. Therefore, we can conclude that $Z\left(\widetilde{W}_2^{(1)}\right)=Z(W)=Z_0$.

Under the C1 mode, due to $\forall y_2, \mathbb{W}\left(y_2\left|0\right.\right)=1,\mathbb{W}\left(y_2\left|1\right.\right)=0$, for $\forall y_1^2$, the transition probabilities of polarized channel $\widetilde{W}_2^{(2)}$ can be written by
\begin{equation}
\left\{
\begin{aligned}
&\widetilde{W}_2^{(2)}\left(y_1^2,u_1\left|0\right.\right)=\frac{1}{2}W(y_1|u_1)\mathbb{W}(y_2|0)=\frac{1}{2}W(y_1|u_1)\\
&\widetilde{W}_2^{(2)}\left(y_1^2,u_1\left|1\right.\right)=\frac{1}{2}W(y_1|u_1\oplus1)\mathbb{W}(y_2|1)=0
\end{aligned}
\right..
\end{equation}
Thus we have $\mathbb{L}\left(y_1^2,u_1\right)=\ln\frac{\widetilde{W}_2^{(2)}\left(y_1^2,u_1\left|0\right.\right)}{\widetilde{W}_2^{(2)}\left(y_1^2,u_1\left|1\right.\right)} =\mathbb{L}(y_2)=+\infty$, which means that $Z\left(\widetilde{W}_2^{(2)}\right)=0$.

% that's all folks

\begin{thebibliography}{99}

\bibitem{Polarcode_Arikan}
E. Ar{\i}kan, ``Channel polarization: a method for constructing capacity achieving codes for symmetric binary-input memoryless channels,'' \emph{IEEE Trans. Inf. Theory}, vol. 55, no. 7, pp. 3051-3073, July 2009.

\bibitem{RCPC_Hagenauer}
J. Hagenauer, ``Rate-compatible punctured convolutional codes (RCPC Codes) and their applications," \emph{IEEE Trans. Commun.}, vol. 36, no. 4, pp. 389-400, April 1988.

\bibitem{RCPT_Rowitch}
D. N. Rowitch and L. B. Milstein, ``On the performance of hybrid FEC/ARQ system using rate compatible punctured turbo (RCPT) codes,"  \emph{IEEE Trans. Commun.}, vol. 48, no. 6, pp. 948-959, June 2000.

\bibitem{polar_rate}
E. Ar{\i}kan and E. Telatar, ''On the rate of channel polarization,'' \emph{in Proc. IEEE Int. Symp. Inform. Theory (ISIT)}, pp. 1493-1495, July 2009.

\bibitem{LP_Goela}
N. Goela, S. B. Korada, and M. Gastpar, ``On LP decoding of polar codes," \emph{Proc. IEEE ITW}, pp. 1-5, Sep. 2010.

\bibitem{SCL_Tal}
I. Tal and A. Vardy, ``List decoding of polar codes,'' \emph{in Proc. IEEE Int. Symp. Inform. Theory (ISIT)}, pp. 1-5, 2011.

\bibitem{SCL_Add}
Y. Fan, C. Xia, J. Chen, C. Tsui, \emph{et al.}, ``A Low-Latency List Successive-Cancellation Decoding Implementation for Polar Codes,'' \emph{IEEE J. Sel. Areas Commun.}, vol. 34, no. 2, pp. 303-317, Feb. 2016.

\bibitem{CASCL_Niu}
K. Niu and K. Chen, ``CRC-aided decoding of polar codes,'' \emph{IEEE Commun. Lett.}, vol. 16, no. 10, pp. 1668-1671, Oct. 2012.

\bibitem{ASCL_Li}
B. Li, H. Shen, and D. Tse, ``An adaptive successive cancellation list decoder for polar codes with cyclic redundancy check," IEEE Commun. Lett., Vol. 16, No. 12, pp. 2044-2047, 2012.

\bibitem{SCS_Niu}
K. Niu and K. Chen, ``Stack decoding of polar codes,'' \emph{Electronics Letters}, vol. 48, no. 12, pp.	695-697, 2012.

\bibitem{SCH_Chen}
K. Chen, K. Niu and J. R. Lin, ``Improved successive cancellation decoding of polar codes,'' \emph{IEEE Trans. Commun.}, vol. 61, no. 8, pp. 3100-3107, 2013.

\bibitem{Survey_Niu}
K. Niu, K. Chen, J. R. Lin, and Q. T. Zhang, ``Polar codes: primary concepts and practical decoding algorithms," \emph{IEEE Commun. Mag.}, pp. 192-203, July 2014.

\bibitem{Practical_Eslami}
A. Eslami and H. Pishro-Nik, ``A practical approach to polar codes,'' \emph{in Proc. IEEE Int. Symp. Inform. Theory (ISIT)}, pp. 16-20, 2011.

\bibitem{Finite_Eslami}
A. Eslami and H. Pishro-Nik, ``On finite-length performance of polar codes: stopping sets, error floor, and concatenated design," \emph{IEEE Trans. Commun.}, vol. 61, no. 3, March 2013.

\bibitem{RCPP_Niu}
K. Niu, K. Chen, and J. R. Lin, ``Beyond turbo codes: rate-compatible punctured polar codes," \emph{in Proc. IEEE International Conference on Communications}, pp. 3423-3427, 2013.

\bibitem{LCPC_Shin}
D. M. Shin, S.-C. Lim, and K. Yang, ``Design of length-compatible polar codes based on the reduction of polarizing matrices," \emph{IEEE Trans. Commun.}, vol. 61, no. 7, pp. 2593-2599, 2013.

\bibitem{Novel_Punc}
R. Wang and R. Liu, ``A novel puncturing scheme for polar codes," \emph{IEEE Commun. Lett.}, vol. 18, no. 12, pp. 2081-2084, 2014.

\bibitem{Shorten_PC}
V. Miloslavskaya, ``Shortened polar codes,"\emph{IEEE IEEE Trans. Inf. Theory}, vol. 61, no. 9, pp. 4852-4865, 2015.

\bibitem{PuncPattern_Zhang}
L. Zhang, Z. Y. Zhang, et. al., ``On the puncturing patterns for punctured polar codes," \emph{in Proc. IEEE Int. Symp. Inform. Theory (ISIT)}, pp. 121-125, 2014.

\bibitem{Weight_distribution}
M. Valipour and S. Yousefi, ``On probablistic weight distribution of polar codes," \emph{IEEE Commun. Lett.}, vol. 17, no. 11, pp. 2120-2123, 2013.

\bibitem{Distance_spectrum}
Z. Z. Liu, K. Chen, K. Niu, and Z. Q. He, ``Distance spectrum analysis of polar codes," \emph{in Proc. IEEE Wireless Commun. Networking Conf. (WCNC) }, pp. 490-495, April 2014.

\bibitem{GA_Trifonov}
P. Trifonov, ``Efficient design and decoding of polar codes,'' \emph{IEEE Trans. Commun.}, vol. 60, no. 11, pp. 3221-3227, Nov. 2012.

\bibitem{Book_Lin}
S. Lin and D. J. Costello Jr., \emph{Error Control Coding: Fundamentals and Applications (2nd ed.)}, Pearson Education, 2004.

\bibitem{Book_Riley}
K. F. Riley, M. P. Hobson, S. J. Bence, \emph{Mathematical methods for physics and engineering (3rd ed.)}, Cambridge University Press, p. 118, 2010.

\bibitem{LTE_36212}
3GPP TS 36.212: ``Multiplexing and channel coding," Release 8, 2009.

\bibitem{CRC_opt}
Q. Zhang, A. Liu, X. Pan and K. Pan, ``CRC Code Design for List Decoding of Polar Codes,'' \emph{IEEE Commun. Lett.}, vol. 21, no. 6, pp. 1229-1232, Jun. 2017.

\bibitem{Qualcomm_Proposal}
3GPP TSG R1-1610137: ``LDPC Rate Compatible Design Overview,'' Qualcomm Incorporated, Lisbon, Portugal, Oct. 2016.

\end{thebibliography}
\end{document}